\documentclass[journal,final]{IEEEtran}

\usepackage{amssymb,amsmath,amsfonts,bm,epsfig,graphicx, epstopdf,wasysym}
\usepackage{rotating,setspace,latexsym,epsf,color,dsfont,caption}
\usepackage{cite,subfigure}
\usepackage{algorithm,algpseudocode}

\usepackage{amsthm}

\usepackage{hyperref}
\hypersetup{
    colorlinks=true,
    linkcolor=blue,
    filecolor=magenta,      
    urlcolor=cyan,
    pdftitle={Overleaf Example},
    pdfpagemode=FullScreen,
    }

\allowdisplaybreaks  

\newtheorem{Definition}{Definition}
\newtheorem{Theorem}{Theorem}

\newtheorem{Lemma}{Lemma}

\theoremstyle{remark}
\theoremstyle{definition}
\newtheorem{remark}{Remark}

\usepackage{amsmath}

\makeatletter
\newcommand\fs@spaceruled{\def\@fs@cfont{\bfseries}\let\@fs@capt\floatc@ruled
  \def\@fs@pre{\vspace{5\baselineskip}\hrule height.8pt depth0pt \kern2pt}%
  \def\@fs@post{\kern2pt\hrule\relax}%
  \def\@fs@mid{\kern2pt\hrule\kern2pt}%
  \let\@fs@iftopcapt\iftrue}
\makeatother

\newcommand{\xv}{\mathbf{x}}

\newcommand{\yv}{\mathbf{y}}

\newcommand{\wv}{\mathbf{w}}
\newcommand{\Wv}{\mathbf{W}}

\newcommand{\Hv}{\mathbf{H}}
\newcommand{\ov}{\mathbf{0}}

\newcommand{\Fv}{\mathbb{F}}

\newcommand{\Rb}{\mathbb{R}}
\newcommand{\Eb}{\mathbb{E}}

\newcommand{\Zb}{\mathbb{Z}}
\newcommand{\Nb}{{\mathbb N}}

\newcommand{\Rc}{\mathcal{R}}

\title{Precoding and Scheduling for AoI Minimization in MIMO Broadcast Channels}
\author{Songtao~Feng, and Jing~Yang~\IEEEmembership{Member,~IEEE}%
\thanks{Songtao~Feng and Jing~Yang are with the School of Electrical Engineering and Computer Science, The Pennsylvania State University, University Park, PA, 16802, USA. Email: \{sxf302, yangjing\}@psu.edu. This work was presented in part in the 2020 IEEE International Symposium on Information Theory~\cite{Feng:ISIT:2020}, and was supported in part by the US National Science Foundation (NSF) under Grants CNS-1956276 and CNS-2114542.\newline
Copyright (c) 2017 IEEE. Personal use of this material is permitted.  However, permission to use this material for any other purposes must be obtained from the IEEE by sending a request to pubs-permissions@ieee.org.}
}

\begin{document}
\IEEEoverridecommandlockouts
\maketitle
\thispagestyle{empty}

\begin{abstract}
In this paper, we consider a status updating system where updates are generated at a constant rate at $K$ sources and sent to the corresponding recipients through a noise-free broadcast channel. We assume that perfect channel state information (CSI) is available at the transmitter before each transmission, and the transmitter is able to utilize the CSI information to precode the updates. Our object is to design optimal precoding schemes to minimize the summed average \emph{age of information} (AoI) at the recipients. Under various assumptions on the size of each update $B$, the number of transmit antennas $M$, and the number of receive antennas $N$ at each user, this paper identifies the corresponding age-optimal precoding and transmission scheduling strategies. Specifically, for the case when $N=1$, a round-robin based updating scheme is shown to be optimal. For the two-user systems with $N>B$ or $M\notin[N:2N]$, framed updating schemes are proven to be optimal. For other cases in the two-user systems, a framed alternating updating scheme is proven to be $2$-optimal.
\end{abstract}


\section{introduction}
Motivated by a variety of network applications requiring timely information, the notion of Age of Information (AoI) is introduced recently \cite{infocom/KaulYG12}. AoI characterizes the freshness of information from the destination's perspective. Specifically, at time $t$, the AoI is defined as the time that has elasped since the latest received update was generated.

A great amount of work focuses on AoI analysis for different queueing models, in which  updates are generated randomly at the source and transmitted to the destination with a random ``service time" through a noiseless channel based on the queueing management model. For single-server systems, the correspondign AoI has been analyzed in the single-source single-server queues~\cite{infocom/KaulYG12}, the $M/M/1$ Last-Come First-Served (LCFS) queue with preemption in service~\cite{ciss/KaulYG12}, the $M/M/1$ queues with multiple sources~\cite{YatesK16,Pappas:2015:ICC,Moltafet:AoIinMM1:2019}, the $M/G/1$ queues~\cite{Yates:AoIinMG1:ISIT:2019,Moltafet:AoI:ITW:2019,Moltafet:AoI:6GSummit:2020}, the $G/G/1$ queues~\cite{Soysal:AoIinGG11:2019,Champati:AoIinGG1:2019}, the LCFS queue with gamma-distributed service time and Poisson update symbol arrivals~\cite{isit/NajmN16}, etc. As for the multiple-server queue, the AoI analysis has been studied in~\cite{isit/KamKE13, isit/KamKE14,tit/KamKNE16,Javani:AoISHS:2019}. For multi-hop networks, the optimality properties of a preemptive Last Generated First Served (LGFS) service discipline are established in~\cite{isit/BedewySS16}, and explicit age distributions based on a stochastic hybrid system approach are derived in~\cite{isit/Yates15}.
For LCFS queue with preemptive service and $G/G/ \infty$ queue, a heavy tailed service time distribution resulting in the worst case symbol delay or variance of symbol delay has been shown to minimize the AoI in~\cite{Talak:AoI:ISIT:2019}. 
With the knowledge of the server state, the AoI optimization has been studied in single-user systems~\cite{infocom/SunUYKS16,Sun:ISIT:2017,Wang:JCN:2019}. It is shown in~\cite{infocom/SunUYKS16} that the zero-wait policy does not always minimize the AoI, while reference~\cite{Sun:ISIT:2017} shows that the age-optimal policy has a threshold structure.

In systems where specific communication channels instead of abstract ``servers" are considered, AoI optimization has also been extensively studied in~\cite{He:2017,Kaul:2017:MAC,Qiao:AoI:WCSP:2019,Najm:AoIinErasure:2019,Javani:AoIinErasure:2019}.
The minimum AoI scheduling problem with interfering links is studied in~\cite{He:2017}. 
The AoI over multiple-access channels has been analyzed for both scheduled access with feedback and slotted ALOHA-like random access mechanisms~\cite{Kaul:2017:MAC}.
Reference \cite{Qiao:AoI:WCSP:2019} investigates the minimization of the average AoI in status update systems with packet based transmissions over fading channels. 
The optimal achievable average AoI over an erasure channel has been studied in~\cite{Najm:AoIinErasure:2019} for the cases when the source and channel-input alphabets have equal or different sizes.
The optimal error toleration policy for AoI minimization during transmission of an update in an erasure channel with feedback has been investigated in~\cite{Javani:AoIinErasure:2019}.

This work investigates broadcast channels similar to those studied in \cite{Modiano:2018:BC,Hsu:2018:ISIT,Sun:AoI_BC:2019,Chen:CodeBC:2019,Feng:Globecom:2019}. Reference~\cite{Modiano:2018:BC} studies the expected weighted sum AoI minimization problem of the single-hop broadcast network with minimum throughput constraints.
It considers a system where the updates for users are generated periodically, and the transmission between the transmitter and each user can be erased with a constant probability. It shows that in a symmetric network, greedily updating the user with the highest instantaneous AoI is optimal. For general setups, it develops low-complexity scheduling policies with performance guarantees.
In \cite{Hsu:2018:ISIT}, it considers stochastic update arrivals while assuming no-buffer transmitter and reliable links between the transmitter and the users. It derives the Whittle's index in a closed-form and proposes a scheduling algorithm based on it. 
Reference~\cite{Sun:AoI_BC:2019} extends results in~\cite{Modiano:2018:BC,Hsu:2018:ISIT} by jointly considering both unreliable links and stochastic update arrivals, and examines Whittle’s index based scheduling policies. A common assumption in \cite{Modiano:2018:BC,Hsu:2018:ISIT,Sun:AoI_BC:2019} is that {\it only one} user can be updated each time. 
Thus, the ``broadcast" nature of wireless medium is not really exploited in those works. 

Recently, a few works have taken some initial steps to explore the benefit of broadcasting on information freshness by relaxing the assumption that only one user can be updated each time~\cite{Chen:CodeBC:2019,Feng:Globecom:2019}. 
In \cite{Chen:CodeBC:2019}, it considers a two-user broadcast {\it symbol} erasure channel with feedback, where a transmitted update can be successfully received by each of the users with certain probability. Based the instantaneous symbol delivery feedback, the transmitter is able to adaptively code the updates and improve the AoI performance of the corresponding uncoded policies. In \cite{Feng:Globecom:2019}, we consider a two-user broadcast {\it symbol} erasure channel, and propose an adaptive coding policy. We show that compared with a greedy transmission policy without coding, the AoI at the weak user can be improved by orders of magnitude without affecting that at the strong user. Both works in ~\cite{Chen:CodeBC:2019,Feng:Globecom:2019} show the benefit of coding on AoI in those broadcast channels.

In this work, we consider a status monitoring system with $K$ sources, each generating updates intended for one of the $K$ recipients. The updates are transmitted to the monitors through a broadcast channel. Different from the models studied in existing works, we consider block fading over the links between the transmitter and receivers, each receiver equipped with $N$ antennas. Therefore, all receivers are able to receive an attenuated version of the transmitted signal. Then, under the assumption that the noise level is negligible in the channel, and the instantaneous channel state information (CSI) is available to the transmitter at the beginning of each time slot, our objective is to investigate the optimal coding and transmission scheduling schemes for the minimization of the summed time-average AoI over the receivers. 

Our main contributions are summarized as follows.

First, we investigate a novel MIMO broadcast setting where optimizing AoI through precoding and transmission scheduling is the focal point. While precoding strategies for throughput optimization for such channel has been investigated extensively in the literature~\cite{Lee:Precoding:MIMO,Caire:TP_BC:TIT,Weingarten:capacity_BC:ISIT}, maximizing information freshness is a very different aspect and requires unconventional treatment. On the other hand, existing study on AoI in broadcast channels rarely considers the impact of multiplexing gain on information freshness. The problem studied in this work bridges the gap between existing studies on MIMO broadcast channel and AoI, rendering novel precoding and transmission scheduling solutions.

Second, we explicitly identify the optimal updating strategies for the MIMO broadcast channel under different setups. Our result indicates that the size of updates plays a critical role on the design of the optimal updating schemes: When updates are of size one, the optimal schemes exhibit a round-robin structure. When updates are of size $B$, $B\geq 2$, round-robin updating may not be optimal. Rather, the transmitter may waste some transmission opportunities in order to deliver fresher updates. This is in contrast to conventional throughput-optimal transmission schemes in the literature. For the two-user case, we show that framed updating schemes are optimal.

Finally, the techniques we develop to show the optimality of the proposed updating schemes are novel. Due to combinatorial nature of the scheduling problem, establishing the optimality of an updating scheme is not straightforward. Toward that, we strive to obtain lower bounds that match with the summed long-term average AoI achieved under the proposed schemes. For the case when each user is equipped with one receiving antenna, we focus on consecutive time frames consisting of $B$ time slots and identify a lower bound on the summed AoI over each frame based on a newly defined notion of Degree of Freedom (DoF). Such DoF characterizes the transmission and decoding capabilities of the system and determines the minimum possible AoI of the users. 
For the two-user case, we first investigate an updating scheme that always update users in an alternating fashion for the timely delivery of each update to the intended user. Such alternating updating schemes naturally partition the time axis into concatenating segments with different updating patterns. We then examine the updating patterns on those segments individually and obtain a lower bound on the corresponding AoI. Finally, we show that the lower bound on the summed long-term average AoI for the class of alternating policies remains valid for any policy.
We believe those techniques are new in the study of AoI, and may be applicable for other problems in the area as well.

\emph{Notation:} Throughout the paper, we use boldface lower case to indicate vectors and boldface upper case to denote matrices. $\Nb$ denotes natural numbers, $\Zb_{\geq n}$ represents integers starting from $n$, and $\Zb_+$ represents integers starting from zero. Besides, we use $[m:n]$ to denote the subset of integers ranging from $m$ to $n$.

\section{Problem Formulation}
Consider a status updating system where there are $K$ independent sources intended for $K$ users. At the beginning of each time slot, an update of $B$ symbols are generated at each source. The symbols are transmitted to the users through a MIMO broadcast channel, as illustrated in Fig.~\ref{fig:model}. We assume the transmitter is equipped with $M$ transmitting antennas, and each receiver is equipped with $N$ receiving antennas. Each user tracks the status of the source of interest based on the symbols it receives. We use $(K,M,N,B)$ to denote the status monitoring system with $K$ source-user pairs, $M$ transmitting antennas, $N$ antennas at each user, and update size $B$.

\begin{figure}[t]
	\centering
	\includegraphics[width=1.8in]{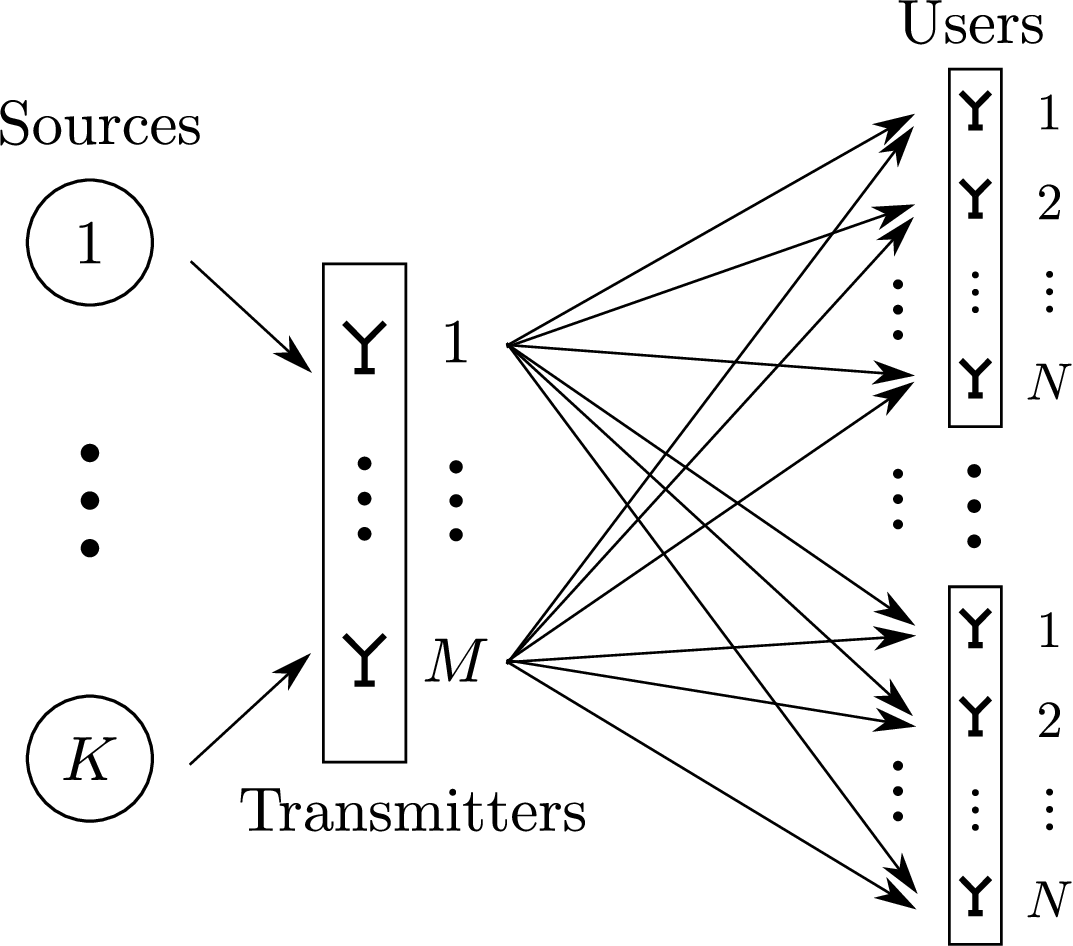}
	\captionof{figure}{System model.} 
	\label{fig:model}
\end{figure}

We refer the $K$ updates generated at time slot $t$ as $\Wv_t:=(\wv_t^{(1)},\ldots,\wv^{(K)}_t)^\mathsf{T}$ where $\wv^{(k)}_t:=(w^{(k)}_t[1],\ldots,w^{(k)}_t[B])^\mathsf{T}$ is the update intended for user $k$. We assume $w_t^{(k)}[b]$ is drawn from a finite field $\Fv_q$, and use $\Wv_t[:,b]$ to denote the $b$-th column of $\Wv_t$. 
We assume at each time slot, the transmitter is able to transmit a symbol on each of its antennas and it takes one time slot to deliver the symbol. Let $\xv_t:=(x_t[1],\ldots,x_t[M])^\mathsf{T}$ be the symbols transmitted at time slot $n$. Throughout this paper, we restrict to linear precoding schemes and assume $\xv_t$ is a {\em linear} function of the previously generated symbols of updates $\{\Wv_\tau\}_{\tau=1}^t$. 

Let $\Hv^{(k)}_t\in(\Fv_q)^{N\times M}$, $k\in[1:K]$, be the channel state between the transmitter and user $k$, and denote $\Hv_t:=((\Hv_t^{(1)})^\mathsf{T},\ldots,(\Hv_t^{(K)})^\mathsf{T})^\mathsf{T}\in(\Fv_q)^{KN\times M}$. The channel output at user $k$, denoted as $\yv^{(k)}_t:=(y^{(k)}_t[1],\ldots,y^{(k)}_t[N])^\mathsf{T}$, is modeled as
\begin{align}
\yv_t^{(k)}=\Hv_t^{(k)}\xv_t,
\end{align}
where we assume the additive noise in the channel is negligible compared with the transmit signal and leave it out for ease of exposition. 

We assume any submatrix of $\Hv_t$ is full rank almost surely, and $\Hv_t$ is available to the transmitter and the users at the beginning of each time slot. Then, the transmitter is able to design $\xv_t$ based on the instantaneous channel state information (CSI) $\Hv_t$, the symbols in $\{\Wv_\tau\}_{\tau=1}^t$ and all previously transmitted symbols $\{\xv_\tau\}_{\tau=1}^{t-1}$. Once $\yv_t$ is received, each individual user $k$ will try to recover updates from the corresponding source $k$ based on received symbols $\{\yv^{(k)}_\tau\}_{\tau=1}^{t}$ and historical CSI $\{\Hv^{(k)}_\tau\}_{\tau=1}^t$.

We adopt the metric \emph{age of information} (AoI) to measure the freshness of the information at the users. Formally, the AoI at user $k$ is the duration since the decoded freshest update was generated at the associated source $k$. Once an intended update is decoded at user $k$, its AoI is reset to the age of the update if it is fresher. If multiple updates from source $k$ are decoded at the same time, the AoI is reset to the age of the freshest one. Let $\delta^{(k)}_t$ be the AoI of the $k$-th user at the end of time slot $n$. Then, the average AoI of user $k$ is defined as
\begin{align}
\Delta^{(k)}=\limsup_{T\rightarrow\infty}\frac{1}{T}\Eb \left[\sum_{t=1}^{T}\delta^{(k)}_t\right],
\end{align}
and the summed average AoI of $K$ users is defined as $\Delta=\sum_{k=1}^{K}\Delta^{(k)}$. 
Our objective is to obtain an optimal precoding policy to determine $\{\xv_t\}_t$, such that the summed average AoI $\Delta$ is minimized.

\section{Main Results}
Due to the combinatorial nature of the precoding and scheduling schemes in the MIMO broadcast channel, searching for the age-optimal updating policy is extremely complicated in general. In order to gain some insights to this general problem, in this paper, we focus on two special scenario. In the first scenario, we restrict to the case when each user is equipped with one receiving antenna, while for the second scenario, we focus on systems with two users only.  
Our main results are summarized as follows.
\begin{Theorem}\label{result1}
For $(K,M,1,B)$ systems, the following results hold:
\begin{itemize}
\item[(i)] If $K\leq M$, the minimum summed average AoI equals $\frac{1}{2}K(3B-1)$;
\item[(ii)] If $K=pM+q$, where $p\in\Nb$ and $q\in[0:M-1]$, the minimum summed average AoI equals $\frac{1}{2}pM(pB+2B-1)+\frac{1}{2}q(2pB+3B-1)$.
\end{itemize}
\end{Theorem}

\begin{Theorem}\label{result2_new}
For $(2,M,N,B)$ systems, the following results hold:
\begin{itemize}
\item[(i)] If $N\geq B$ and $\frac{M}{B}\geq 2$, the minimum summed average AoI equals 2;
\item[(ii)] If $N\geq B$ and $1\leq \frac{M}{B}< 2$, the minimum summed average AoI equals 3;
\item[(iii)] If $N\geq M$, $0<\frac{M}{B}<1$, let $i=\lceil \frac{B}{M}\rceil -1$ and $j=\lfloor \frac{1}{B/M-i}\rfloor$. Then, $\frac{j}{ij+1}\leq \frac{M}{B}< \frac{j+1}{(j+1)i+1}$, and the minimum summed average AoI equals $4i+1+\frac{2i+1}{ij+1}$ if $j\geq 2$, and equals $4i+3$ if $j=1$.
\item[(iv)] If $N\leq \frac{M}{2}$, $0<\frac{N}{B}<1$, let $i=\lceil \frac{B}{N}\rceil -1$ and $j=\lfloor \frac{1}{B/N-i}\rfloor$. Then, $\frac{j}{ij+1}\leq \frac{N}{B}<\frac{j+1}{(j+1)i+1}$, and the minimum summed average AoI equals $3i+1+\frac{i+1}{ij+1}$.
\end{itemize}
\end{Theorem}

We note that Theorem~\ref{result2_new} explicitly characterizes the optimal AoI in all $(2,M,N,B)$ systems except for the case when $N<B$ and $N<M<2N$. Although explicit identification of the optimal AoI for this case is extremely challenging and intractable, we are able to provide a lower bound on the summed average AoI, and obtain performance guarantee for a transmission policy as follows.

\begin{Theorem}\label{result3_new}
For $(2,M,N,B)$ systems, if $B=iN+j$, $i\in\Nb$, $j=B\pmod N$, and $N<M<2N$, the minimum summed average AoI is lower bounded by $\Delta_{\mathsf{LB}}=2i+\lceil\frac{2j}{N}\rceil$. Moreover, there exists a $2$-optimal policy under which the summed average AoI is upper bounded by $2\Delta_{\mathsf{LB}}$.
\end{Theorem}

In the following, we first present updating schemes in Section~\ref{sec:achieve-thm1} and Section~\ref{sec:achievale-thm2} that achieve the summed time-average AoI in Theorem~\ref{result1} and Theorem~\ref{result2_new}, and then provide the matching lower bounds in Section~\ref{sec:converse-thm1}  and Section~\ref{sec:converse-thm2}. In Section~\ref{sec:discuss}, for $(2,M,N,B)$ systems with $N<B$ and $N<M<2N$, we investigate the lower bound and propose a $2$-optimal policy. We conclude the paper in Section~\ref{sec:conclusion} and defer some of the proofs to the Appendix.

\section{Achievable Schemes for Theorem~\ref{result1}}\label{sec:achieve-thm1}
In this section, we explicitly describe the optimal updating schemes that render the minimum summed average AoI stated in Theorem~\ref{result1}.

\subsection{Achievable Scheme for the $(K,M,1,B)$ System with $M\geq K$}\label{sec:achievale-thm1-case1}
First, we consider the case when $N=1$, $M\geq K$. This corresponds to the case when each user is equipped with a single antenna, and the number of antennas at the transmitter is greater than the number of users. Since $N=1$, each user can receive at most one linear combination of the transmitted symbols, implying that a $B$-symbol update takes at least $B$ time slots to deliver. On the other hand, since $M\geq K$, the transmitter is able to send $K$ independent symbols in each time slot. This motivates us to propose a simple synchronized updating scheme as follows:

\begin{Definition}[Synchronized updating]
Partition the time axis into frames of length $B$ starting at the beginning of time slot 1. Then, at the beginning of time slot $t=mB+b$, $m\in\Zb_+$, $b\in[1:B]$. The transmitter sends
\begin{align}
\xv_t=
\begin{pmatrix} 
\Tilde{\Hv}_t^{-1}\Wv_{mB+1}[:,b]   \\
\mathbf{0}_{(M-K)\times 1} 
\end{pmatrix},  \label{eqn:achieve-thm1-case1:1}
\end{align}
where $\Tilde{\Hv}_t\in(\Fv_q)^{K\times K}$ is $\Hv_t$ knocked off the last $M-K$ columns.
\end{Definition}

We note that under the synchronized updating scheme, the $b$-th symbol of updates generated at the beginning of a time frame, i.e., $\{\Wv_{mB+1}^{(k)}\}_k$, is transmitted in the $b$-th time slot in the corresponding time frame simultaneously. By precoding the symbols according to (\ref{eqn:achieve-thm1-case1:1}), each user is able to cancel off the interference from other unintended updates and decode the designated update at the end of the time frame, i.e., at the end of time slot $(m+1)B$. The synchronized updating scheme for the $(3,4,1,2)$ system is shown in Fig.~\ref{fig:thm1-case1}.

\begin{figure}[t]
	\centering
	\includegraphics[height=1.5in]{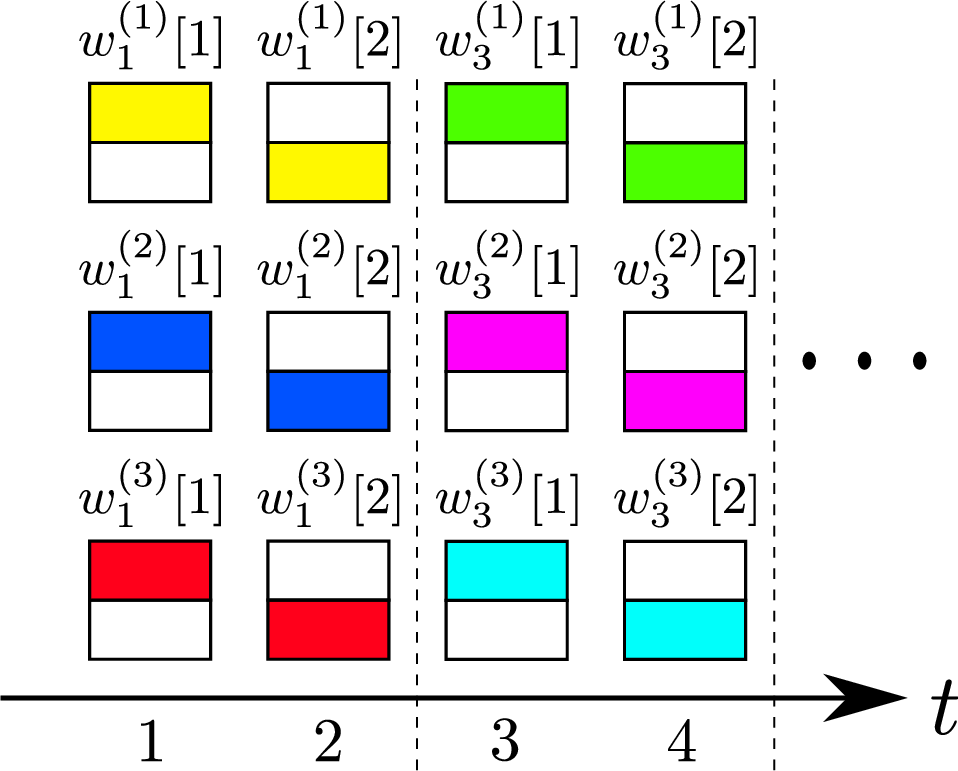}
	\captionof{figure}{Synchronized updating scheme for the $(3,4,1,2)$ system, where updates $\Wv_1,\Wv_3,\ldots$ are delivered at the end of time slots $2,4,\ldots$. }
	\label{fig:thm1-case1}
\end{figure}

Tracking the AoI of each user $\delta^{(k)}_t$ in time frame consisting of time slots $[mB+1:(m+1)B]$, $m\in\Nb$, we note that for general $B>1$, it increases monotonically from $B+1$ to $2B-1$ until being reset to $B$ at the end of the time frame. When $B=1$, the AoI resets to 1 at the end of each time slot. Denote $\delta^{(k)}_{m:n}=\sum_{t=m}^n \delta^{(k)}_t$.  
Assume the initial AoI at time 0 is bounded for every user. Then, 
\begin{align}
\Delta&=\lim_{T\rightarrow\infty}\frac{K}{TB}\sum_{t=1}^{TB}\delta_t^{(k)}  \nonumber\\
&=\lim_{T\rightarrow\infty}\frac{K}{TB}\left(\delta^{(k)}_{1:B}+\sum_{m=1}^{T-1}\delta^{(k)}_{mB+1:(m+1)B}\right)    \nonumber \\
&=\frac{1}{2}K(3B-1). \label{eqn:achieve-thm1-case1:6}
\end{align}
We note that the synchronized updating scheme is not the only updating scheme that achieves the AoI depicted in Theorem~\ref{result1} for the $M\geq K$ case. Actually, instead of starting the transmission of new updates to all users synchronously, the transmitter can continuously update the users in an asynchronous way by introducing an offset $n_k\in[0:B-1]$ to the time when the transmitter starts transmitting a new update to user $k$. Such an offset will only affect the updating time points of a user without changing the AoI evolution pattern between two updates. Thus, the long-term average AoI stays the same.

\subsection{Achievable Scheme for the $(K,M,1,B)$ System with $M< K$}\label{sec:achievale-thm1-case2}
Next, we consider the case when $N=1$, $M< K$. Compared with the scenario discussed in Sec.~\ref{sec:achievale-thm1-case1}, we note that the number of transmitting antennas is now less than the number of users, which implies that not all users can be updated in a synchronized fashion. How to schedule the updating of each user to minimize the total AoI thus becomes non-trivial. We propose the following intuitive updating scheme and prove it is optimal afterwards.

\begin{Definition}[Round-robin synchronized updating] Partition the time axis into frames of length $B$ starting at the beginning of time slot 1.  Then, the transmitter selects $M$ users to update in each frame in a round-robin fashion. Specifically, in the frame consisting of time slots $[mB+1:(m+1)B]$, $m\in\Zb_+$, the selected users to update are the $M$ users $(mM+i-1) \pmod K+1$, $i=[1:M]$. At the beginning of time slot $t=mB+b$, $b\in[1:B]$, the transmitter sends
$\xv_t=\Tilde{\Hv}_t^{-1}\Tilde{\Wv}_{mB+1}[:,b]$,
where $\Tilde{\Hv}_t\in(\Fv_q)^{M\times M}$ and $\Tilde{\Wv}_{mB+1}\in(\Fv_q)^{M\times 1}$ are $\Hv_t$ and  ${\Wv}_{mB+1}$ knocked off the rows associated with the unselected $K-M$ users, respectively.
\end{Definition}

Under the precoding and transmission scheme, the selected $M$ users are able to decode the intended updates at the end of each time frame. 
The transmission strategy for the $(3,2,1,2)$ system is given in Fig.~\ref{fig:thm1-case2} as an example. 

\begin{figure}[t]
	\centering
	\includegraphics[height=1.5in]{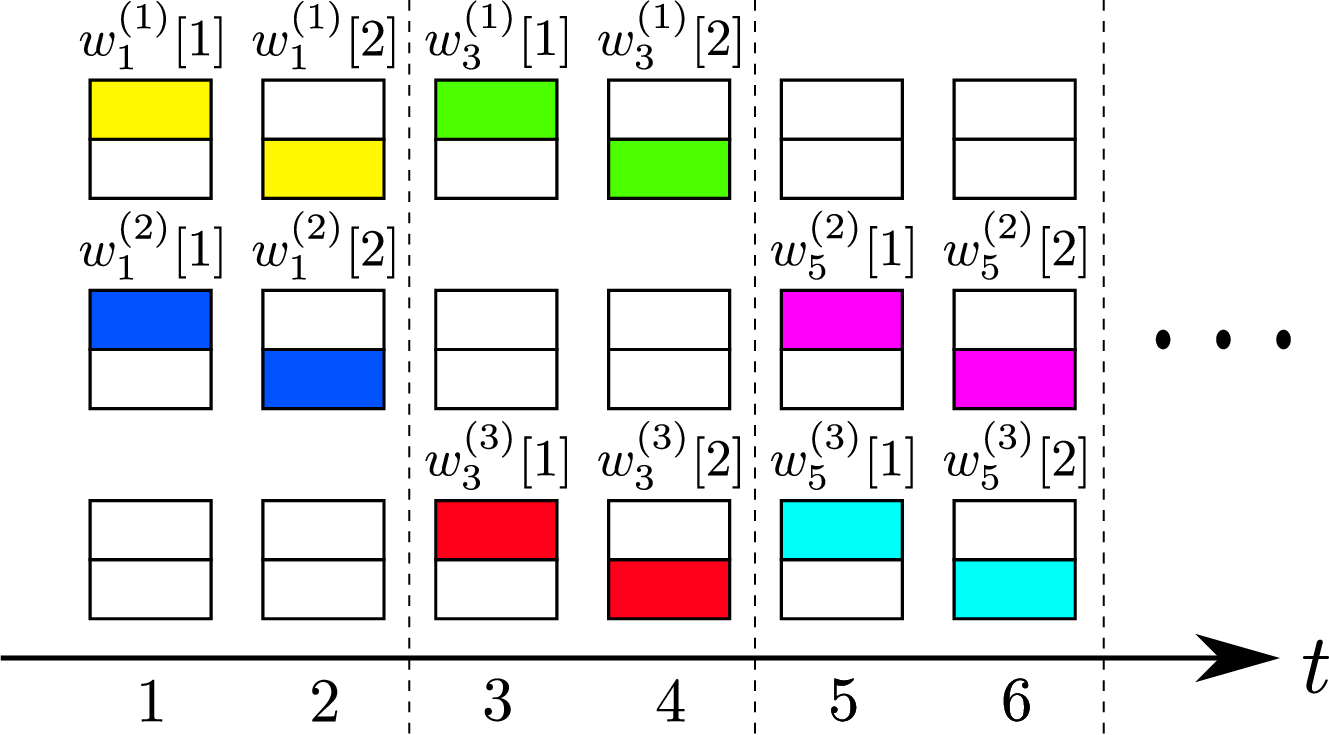}
	\captionof{figure}{Round-robin synchronized updating for the $(3,2,1,2)$ system.} 
	\label{fig:thm1-case2}
\end{figure}

In the following, we explicitly identify the summed time-average AoI under the round-robin updating scheme. 

\begin{Lemma}\label{lemma:rank}
Let $K=pM+q$, where $q=K \pmod M$. Let $k$ be the index of the $i$-th ranked user in the frame consisting of time slots $[mB+1:(m+1)B]$, $m\in\Zb_+$, i.e., $k:=(mM+i-1)\pmod K +1$. Denote $L_i$ as the frame index difference between the current frame and the next frame during which user $k$ will be updated. Then, $L_i= p$ if $i+q\leq M$, and $L_i=p+1$ if $i+q> M$. Besides, user $k$ will be the $(\overline{i+q})$-th ranked user in the frame starting at $(m+L_i)B+1$, where $\overline{x}:=(x-1)\pmod M+1$. 
\end{Lemma}
\begin{proof}
Under the round-robin updating policy, all the rest $K-1$ users should be updated exactly once between two consecutive updates of user $k$. Therefore, we must have $L_i M<i+K\leq (L_i+1) M$. Since $K=pM+q$, the inequality becomes $L_i M<i+pM+q\leq (L_i+1) M$. Therefore, if $i+q\leq M$, we must have $L_i=p$; Otherwise, $L_i=p+1$. Meanwhile, we note that the ranking of user $k$ will be $(i+K-1)\pmod M+1$ in frame $m+L_i$, i.e., $\overline{i+q}$.
\end{proof}

We point out that under the round-robin synchronized updating scheme, each update takes exactly $B$ time slots to transmit. Thus, the AoI at user $k$ after each update is $B$, and it monotonically increases until the next updating time point.

\begin{Lemma}\label{lemma:group}
Let $d:=\gcd(q,M)$, $k:=(mM+i-1)\pmod K +1$. Then, after the frame starting at $mB+1$, the ranking of user $k$ in the next $M/d$ time frames during which user $k$ is updated must be a permutation of $\Rc_i:=\left\{\overline{i+d},\overline{i+2d},\ldots, \overline{i+\frac{M}{d}d}\right\}$. 
\end{Lemma}
\begin{proof}
First, we note that $\overline{i+\frac{M}{d}d}=i$. Thus, for any $\ell \in \Nb$, $\overline{i+\ell d}$ must belong to $\Rc_i$.

Next, for any $\ell \in\Nb$, since $d=\gcd(q,M)$, $\frac{\ell q}{d}\in\Nb$, we must have $\overline{i+\ell q}=\overline{i+\frac{\ell q}{d}d}$. Thus, $\overline{i+\ell q}$ belong to $\Rc_i$, too. 

Besides, for any $1\leq \ell_1<\ell_2\leq M/d$, $\ell_1,\ell_2\in\Nb$, we can show that  $\overline{i+\ell_1 q}\neq \overline{i+\ell_2 q}$ through contradiction as follows: if $\overline{i+\ell_1 q}= \overline{i+\ell_2 q}$, we must have $(\ell_2-\ell_1)q$ be an integer multiple of $M$, i.e., $(\ell_2-\ell_1)q/d$ must be an integer multiple of $M/d$. Since $d=\gcd(q,M)$, it implies $\ell_2-\ell_1$ must be an integer multiple of $M/d$, which contradicts with the assumption that $1\leq \ell_1<\ell_2\leq M/d$.

Therefore, for $\overline{i+\ell q}$, $\ell=1,\ldots, M/d$, they must equal $M/d$ different values, which implies that the ranking of user $k$ in $M/d$ consecutive updating frames must be a permutation of $\Rc_i$.
\end{proof}

\begin{remark}\label{remark:lemma:group}
We note that for two users $k_1,k_2$, $k_1\neq k_2$, if $\bar{k}_1=\bar{k}_2$, they share the same set of rankings when updated. In total, there exist $M$ different set of rankings $\{\Rc_i\}_{i=1}^{M}$.
\end{remark}

Consider $K/d$ consecutive frames. Since $M$ users are updated in each frame, and the updating is performed in a round-robin fashion, each user is updated exactly $M/d$ times. Therefore,  the AoI evolution is periodic every $K/d$ frames after the first update for each user. The long-term average AoI of any user is thus equal to the average AoI during any $K/d$ consecutive frames after its first update. 

Consider the AoI evolution of user $k$ after its first update. We note that under the round-robin synchronized updating scheme, the ranking of user $k$ when it is updated for the first time is $(k-1)\pmod{M}+1$, i.e.,  $\overline{k}$. Consider the $K/d$ consecutive frames starting at time $t=\frac{K}{d}B+1$. According to Lemma~\ref{lemma:group}, we have
\begin{align}
 \Delta^{(k)} = \frac{d}{KB}\sum_{t=\frac{K}{d}B+1}^{\frac{2K}{d}B}\delta^{(k)}_t =\frac{d}{KB} \sum_{j\in \Rc_{\overline{k}}} f(L_j),
\end{align}
where $f(L_j):=\frac{1}{2}[B+B(1+L_j)-1]BL_j$ is the total AoI experienced by user $i$ between two consecutive updates.

Thus,
\begin{align}
\Delta&= \sum_{k=1}^K\Delta^{(k)} =\frac{d}{KB}\sum_{k=1}^K \sum_{j\in \Rc_{\overline{k}}} f(L_j).  \label{eqn:group}
\end{align}
We note that $\{\Rc_{\bar{k}}\}_{k=1}^{K}$ actually corresponds to the rankings of the $K$ users during any consecutive $\frac{K}{d}$ frames when they are updated. Since there are always $M$ users selected in each frame, $\{\Rc_{\bar{k}}\}_{k=1}^{K}$ must contain $M$ different elements from $1$ to $M$, and each element appears exactly $\frac{K}{d}$ times. Applying this observation on Eqn.~(\ref{eqn:group}), we have
\begin{align}
\Delta&=\frac{d}{KB}\sum_{k=1}^K \sum_{j\in \Rc_{\overline{k}}} f(L_j)
=\frac{d}{KB} \sum_{i =1}^M \frac{K}{d} f(L_i)\\
&= \frac{1}{B} \left[(M-q) f(p) + q  f(p+1)\right]\label{eqn:decompose}\\
&=\frac{Mp}{2}(Bp+2B-1)+\frac{q}{2}(2Bp+3B-1).
\end{align}

\section{Achievable Schemes for Theorem~\ref{result2_new}}\label{sec:achievale-thm2}
In this section, we investigate achievable schemes matching the minimum summed average AoI in Theorem~\ref{result2_new}. For those cases, we first focus on the $(2,M,B,B)$, $(2,M,M,B)$ and $(2,2N,N,B)$ systems, respectively, and then show that the corresponding schemes can be applied to systems with general parameter setups.

\subsection{Achievable Scheme for $(2,M,B,B)$ Systems with $M/B\geq 2$}\label{sec:achievale-thm2-case1}
Since $M/B\geq 2$, the transmitter is able to send at least $2B$ linear combinations of update symbols in each time slot. Therefore, at each time slot $t$, the transmitter chooses to transmit all $2B$ symbols of the newly generated updates $\wv_t^{(1)}$ and $\wv_t^{(2)}$. The precoding procedure is as follows: We knock off the last $M-2B$ columns of $\Hv_t^{(1)}, \Hv_t^{(2)}\in(\Fv_q)^{B\times M}$ and let the remaining matrices be $\Tilde{\Hv}_t^{(1)}, \Tilde{\Hv}_t^{(2)}\in(\Fv_q)^{B\times 2B}$. Denote $\tilde{\Hv}_t:=((\Tilde{\Hv}_t^{(1)})^\mathsf{T},(\Tilde{\Hv}_t^{(2)})^\mathsf{T})^\mathsf{T}\in(\Fv_q)^{2B\times 2B}$. At the beginning of time slot $t$, the transmitter selects
\begin{align}
\xv_t=
\begin{pmatrix} 
\Tilde{\Hv}_t^{-1}
\begin{pmatrix}
\wv_t^{(1)} \\
\wv_t^{(2)}
\end{pmatrix}\\
\mathbf{0}_{(M-2B)\times 1}
\end{pmatrix}.
\end{align}

Both users are able decode the intended update at the end of each time slot $t$, resetting the AoI to $1$. Thus, the summed average AoI at the end of each time slot is $2$.

\subsection{Achievable Scheme for $(2,M,B,B)$ Systems with $1\leq M/B< 2$}\label{sec:achievale-thm2-case2}
Next, we consider the scenario when $1\leq M/B< 2$. Since $M<2B$, the two newly generated updates can not be delivered in the same time slot simultaneously. On the other hand, since $M\geq B$, it indicates that at least one update can be delivered in each time slot. Thus, the question becomes whether the transmitter should utilize the remaining transmission capability to transmit another update partially. It turns out that a scheme that updates the two users alternately, one in each time slot, is optimal.

Specifically, at time slot $t$, the transmitter sends
\begin{align}
\xv_{t}=
\begin{pmatrix}
(\Tilde{\Hv}^{(k)}_t)^{-1}\wv^{(k)}_t \\
\ov_{(M-B)\times 1}
\end{pmatrix},
\end{align}
where $k=1$ if $t$ is odd, and $k=2$ if $t$ is even, and $\Tilde{\Hv}_t\in(\Fv_q)^{B\times B}$ is $\Hv^{(k)}_t$ knocked off the last $M-B$ columns.

Then, at the end of time slot $t$, the transmitted update is decoded at the corresponding user. Since the AoI of each user resets to 1 every two time slots, the summed time-average AoI is $3$.

\subsection{Achievable Scheme for $(2,M,M,B)$ Systems with $\frac{j}{ij+1}\leq \frac{M}{B}< \frac{j+1}{(j+1)i+1}$, $i,j\in\Nb$}\label{sec:achievale-thm2-case3}

Consider the case when $\frac{M}{B}<1$. We partition the range $(0,1)$ into intervals in the form of $[\frac{j}{ij+1},\frac{j+1}{(j+1)i+1})$ for $i=\lceil \frac{B}{M}\rceil -1$ and $j=\lfloor \frac{1}{B/M-i}\rfloor$, and construct an achievable scheme for each possible interval that $\frac{M}{B}$ may lie in.

\begin{Definition}[Framed alternating updating]  Partition the time axis into frames of length $ij+1$ starting at time slot 1. Then, the transmitter exhausts its transmission capability to update the two users alternatively until the end of the frame. Specifically, in the frame starting at $m(ij+1)+1$, let $t_n:=m(ij+1)+ni+1$, $n\in[0:j-1]$, i.e., the time slot during which a new update will be transmitted, and $k_n:=(mj+n)\pmod{2}+1$, i.e., the user that the new update is intended to. Then, when $t=t_n$, $n\in[1:j-1]$, 
\begin{align}
\xv_t=
\Tilde{\Hv}_t^{-1}
\begin{pmatrix}
\wv_{{t_{n-1}}}^{(k_{n-1})}[niM-(n-1)B+1:B]\\
\wv_{t_n}^{(k_n)}[1:(ni+1)M-nB]
\end{pmatrix},
\end{align}
where $\Tilde{\Hv}_t\in(\Fv_q)^{M\times M}$ is the channel matrix between the $M$ transmitting antennas and subsets of antennas at users $k_{n-1}$ and $k_n$, respectively. When $t=t_{n}+l$, $n\in[0:j-1]$, $l\in[1:i-1]$, and $t=t_0$, $\xv_t=\tilde{\Hv}_t^{-1} \wv_{t_{n}}^{(k_{n})}[(ni+l)M-nB+1:(ni+l+1)M-nB]$.
\end{Definition}

An example of the framed alternating updating policy for the $(2,7,12,12)$ system is shown in Fig.~\ref{fig:thm2-case5}.

Next, we track the AoI evolution under the framed alternating updating scheme. First, we note that under the proposed scheme, in each frame, the transmitter sends $M$ symbols in each time slot until $j$ updates are delivered to the two users alternately. Besides, 
since $\frac{nB}{M}\in[(i+\frac{1}{j+1})n,(i+\frac{1}{j})n)$, when $n\in[1:j]$, the $n$-th updating time in the frame starting at time slot $(ij+1)m+1$ must be $t_n=(ij+1)m+ni+1$. Moreover, since $nB<M(i+1)n$, there must be some transmission capability left after delivering the update in time slot $t_n$ for $n<j$, which will be used to transmit a new one. Therefore, once the update is delivered at time $t_{n+1}=(ij+1)m+i(n+1)+1$, the corresponding AoI is reset to $i+1$. 

\begin{figure}[t]
	\centering
	\includegraphics[width=3.4in]{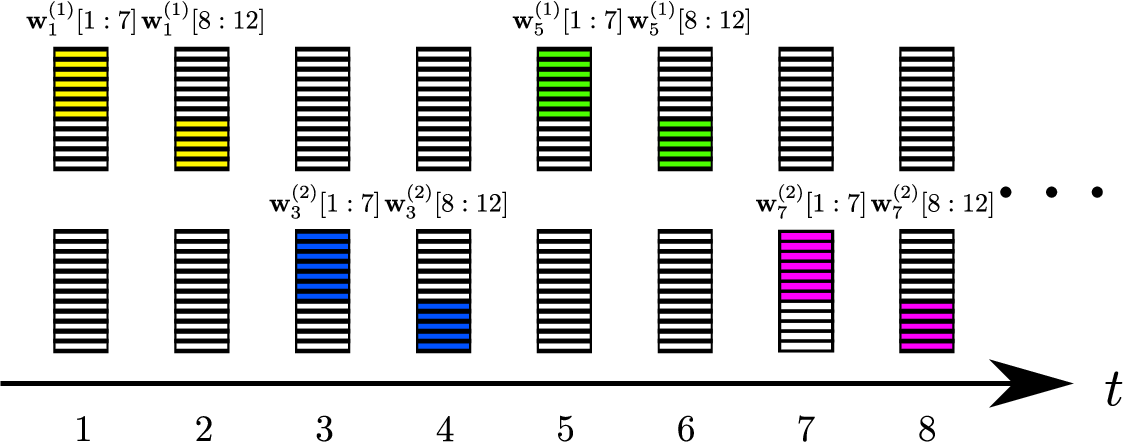}
	\captionof{figure}{Framed alternating updating scheme for the $(2,7,12,12)$ system. Since $\frac{M}{B}=\frac{7}{12}\in[\frac{1}{2},\frac{2}{3})$, we have $i=1,j=1$, and the frame length is 2. Although the transmitter is able to deliver 7 linearly independent symbols in each time slot, it chooses to deliver 5 instead of 7 symbols in the last time slot of each frame.}
	\label{fig:thm2-case5}
\end{figure}

We track the AoI of both users under the updating scheme, and have the following observations.

{\it 1) $j$ is even.} For this case, the AoI evolution is periodic with period $ij+1$. Within each period,  $\delta_{t}^{(1)}$ begins with $2i+2$ and resets to $i+1$ when $t=t_1,\ldots,t_{(j-2)/2}$ and monotonically increases in between; while $\delta_{t}^{(2)}$ begins with $i+2$ and resets to $i+1$ when $t=t_2,\ldots,t_{j/2}$ and monotonically increases in between. The summed long-term average AoI equals the summed AoI over any frame after the first one. 
Therefore,
\begin{align}
\Delta^{(1)}&=\frac{1}{ij+1}\left(\sum_{\ell=0}^{t_1-t_0}(2i+1+\ell) \right.\nonumber\\
&\qquad\qquad \left.+\left(\frac{j}{2}-1\right)\sum_{\ell=1}^{2i}(i+\ell)+\sum_{\ell=1}^{i}(i+\ell)\right)  \\
&=\frac{1}{2}\left(4i+1+\frac{2i+1}{ij+1}\right).
\end{align}
Similarly, we can obtain $\Delta^{(2)}$, which equals $\Delta^{(1)}$. Combining them together, we have $\Delta=4i+1+\frac{2i+1}{ij+1}$.

{\it 2) $j$ is odd.} For this case, under the framed alternating updating policy, one user will be updated $\frac{j-1}{2}$ times, while the other one will be updated $\frac{j+1}{2}$ times within each frame. Thus, the first user to update in next frame will be switched correspondingly. The AoI evolution is periodic with period $2(ij+1)$. Following similar analysis as for the previous case, we can show that $\Delta=4i+1+\frac{2i+1}{ij+1}$ if $j\geq 3$. 

If $j=1$, 
\begin{align}
\Delta^{(1)}=\Delta^{(2)}=\frac{1}{2(i+1)}\sum_{\ell=0}^{2i+1}(i+1+\ell)=\frac{4i+3}{2}.
\end{align}
Thus, $\Delta=4i+3$ if $j=1$.

\subsection{Achievable Scheme for $(2,2N,N,B)$ Systems with $\frac{j}{ij+1}\leq \frac{N}{B}<\frac{j+1}{(j+1)i+1}$, $i,j\in\Nb$}\label{sec:achievale-thm2-case4}

Finally, we consider the case when $\frac{N}{B}<1$ for $(2,2N,N,B)$ Systems. For $\frac{N}{B}\in[\frac{j}{ij+1}, \frac{j+1}{(j+1)i+1})$, an achievable scheme can be constructed similar to that in Section~\ref{sec:achievale-thm2-case3} as follows.

\begin{Definition}[Framed synchronous updating]  Partition the time axis into frames of length {$ij+1$} starting at time slot 1. Then, the transmitter exhausts its transmission capability to update the two users simultaneously until the end of the frame. Specifically, in the frame starting at $m(ij+1)+1$, let $t_n:=m(ij+1)+ni+1$, $n\in[0:j-1]$, i.e., the time slot during which a new update will be transmitted. Then, when $t=t_n$, $n\in[1:j-1]$, 
\begin{align}
\xv_t=
\Tilde{\Hv}_t^{-1}
\begin{pmatrix}
\wv_{{t_{n-1}}}^{(1)}[niN-(n-1)B+1:B]\\
\wv_{{t_{n-1}}}^{(1)}[1:(ni+1)N-nB]\\
\wv_{{t_{n}}}^{(2)}[niN-(n-1)B+1:B]\\
\wv_{{t_{n}}}^{(2)}[1:(ni+1)N-nB] 
\end{pmatrix},
\end{align}
where $\Tilde{\Hv}_t\in(\Fv_q)^{2N\times 2N}$ is the channel matrix between the $M$ transmitting antennas and subsets of antennas at users $k_{n-1}$ and $k_n$, respectively. When $t=t_{n}+l$, $n\in[0:j-1]$, $l\in[1:i-1]$, and $t=t_0$, 
\begin{align}
\xv_t=
\Tilde{\Hv}_t^{-1}
\begin{pmatrix}
\wv_{t_{n}}^{(1)}[(ni+l)N-nB+1:(ni+l+1)N-nB]\\
\wv_{t_{n}}^{(2)}[(ni+l)N-nB+1:(ni+l+1)N-nB]
\end{pmatrix}.
\end{align}

\end{Definition}
An example of the framed synchronous updating policy for the $(2,14,7,12)$ system is shown in Fig.~\ref{fig:thm2-case++}.

\begin{figure}[t]
	\centering
	\includegraphics[width=3.4in]{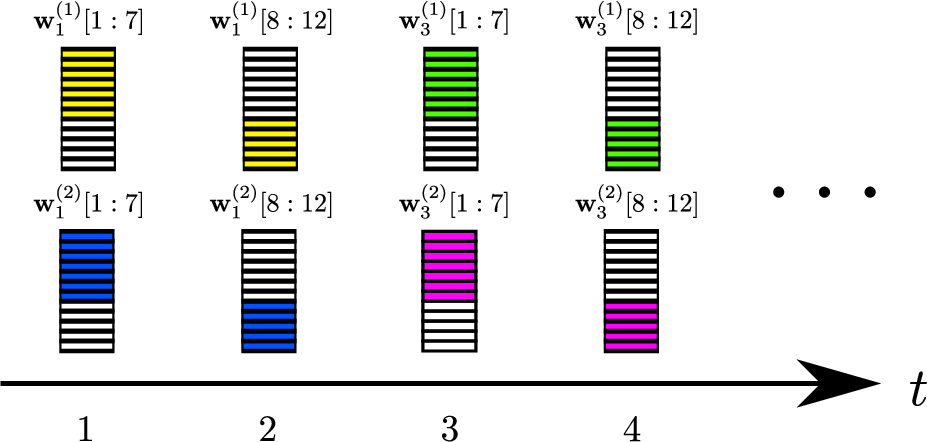}
	\captionof{figure}{Framed synchronous updating scheme for the $(2,14,7,12)$ system. Since $\frac{N}{B}=\frac{7}{12}\in[\frac{1}{2},\frac{2}{3})$, we have $i=1,j=1$, and the frame length is 2. Although the transmitter is able to deliver 7 linearly independent symbols in each time slot, it chooses to deliver 5 instead of 7 symbols in the last time slot of each frame.}
	\label{fig:thm2-case++}
\end{figure}

Different from the framed alternating updating scheme in Section~\ref{sec:achievale-thm2-case3}, the AoI evolutions of two users are always the same under the framed synchronous updating scheme. Following a similar argument as in Section~\ref{sec:achievale-thm2-case3}, we can show that the $n$-th updating time of each user in the frame starting at time slot $(ij+1)m+1$ must be $t_n=(ij+1)m+ni+1$ and the AoI of the delivered update at time $t_{n+1}=(ij+1)m+(n+1)i+1$ must be $i+1$. 

Tracking the AoI of both users, the AoI evolution of each user is periodic with period $ij+1$. Within each period, both $\delta_t^{{1}}$ and $\delta_t^{(2)}$ begins with $2i+2$ and resets to $i+1$ when $t=t_1,t_2,\ldots,t_{j-1}$ and monotonically increases in between. The summed long-term average AoI equals the summed AoI over any frame after the first one. Therefore,
\begin{align}
&\Delta^{(1)}=\Delta^{(2)} \nonumber \\
&=\frac{1}{ij+1}\left(\sum_{\ell=0}^{t_1-t_0}(i+1+\ell)+(j-1)\sum_{\ell=1}^{i}(i+\ell)\right) \\
&=\frac{3i+1}{2}+\frac{i+1}{2(ij+1)}.
\end{align}
Thus, $\Delta=3i+1+\frac{i+1}{ij+1}$.

\subsection{Generalization to $(2,M,N,B)$ Systems}

For $(2,M,N,B)$ systems with $N\geq B$, we note that all updating schemes described for the $(2,M,B,B)$ systems in Sections~\ref{sec:achievale-thm2-case1} and \ref{sec:achievale-thm2-case2} are still applicable. This is equivalent to virtually removing $N-B$ antennas at each receiver, and the corresponding AoI evolution remains the same. Therefore, those updating schemes achieve the corresponding optimal summed average AoI specified in Theorem~\ref{result2_new} (i)-(ii).

Likewise, updating schemes for $(2,M,M,B)$ systems described in Section~\ref{sec:achievale-thm2-case3} can be applied to $(2,M,N,B)$ systems with either $N\geq B>M$ or $B>N\geq M$, while the updating schemes for $(2,2N,N,B)$ systems proposed in Section~\ref{sec:achievale-thm2-case4} can be applied to $(2,M,N,B)$ system with $M\geq 2N$. Therefore, the proposed updating schemes achieve the optimal summed average AoI specified in Theorem~\ref{result2_new} (iii) and (iv), respectively.

\section{Converse of Theorem~\ref{result1}}\label{sec:converse-thm1}
In this section, we prove the converse of Theorem~\ref{result1}, i.e., we will show that the summed long-term average AoI under any updating scheme cannot be lower than that specified in Theorem~\ref{result1}. Towards that, we first introduce the concept of degree of freedom (DoF) in this context, and then define a subset of schemes where the optimal scheme must lie in.

\begin{Definition}[Degree of Freedom (DoF)]\label{defn:DoF}
In a time slot, the degree of freedom (DoF) for a $(K,M,N,B)$ system is the number of linearly independent equations that are delivered to users in the time slot, while the DoF allocated to a user is the number of linearly independent equations that the user receives in the time slot.
\end{Definition}
The definition of DoF characterizes the transmission capability of the system: The total number of symbols decoded by a user cannot exceed the maximum number of linearly independent equations it can receive, as it needs $n$ linearly independent equations to solve for the $n$ unknown variables (symbols).  
For a $(K,M,N,B)$ system, the maximum DoF for each user in any time slot is $\min\{M,N\}$, while the maximum DoF for the whole system is $\min\{M,NK\}$.

\begin{Lemma}\label{lemma:raw}
For any updating scheme, there always exists an equivalent updating scheme under which the precoding matrix is designed in such a way that the receivers' antennas receive raw symbols of the intended updates only and the DoF allocation remains the same. 
\end{Lemma}
\begin{proof}
Without loss of generality, we assume the transmitter starts to update the users at time slot 1. Let $n_k$
be the DoF allocated to user $k$ under the original scheme at time slot 1. Then, we must have $n_k\leq N$, $\tilde{N}:=\sum_{k=1}^{K}n_k\leq M$. Let $\yv^{(k)}_1=(w_1^{(k)}[1],\ldots,w_1^{(k)}[n_k])^\mathsf{T}$ be the symbols received by user $k$ at time 1 under the equivalent updating scheme. We can always design a precoding matrix in the form of
\begin{align}
\tilde{\Hv}_1^{-1}:=
\begin{pmatrix}
\tilde{H}_1^{(1)} \\
\tilde{H}_1^{(2)} \\
\vdots \\
\tilde{H}_1^{(K)}
\end{pmatrix}^{-1},
\end{align}
where $\tilde{H}_1^{(k)}\in(\Fv_q)^{n_k\times \tilde{N}}$ is a submatrix of $H_1^{(k)}$ corresponding to the CSI between the first $\tilde{N}$ transmitting antennas and the first $n_k$ receiving antennas at user $k$. Under the assumption that any submatrix of $\Hv_t$ is full-rank almost surely, $\yv^{(k)}_1$ will be delivered to user $k$ in time slot $1$. We note that the new updating scheme maintains the same DoF allocating under the original scheme. We then continue this process in time slot $2$, during which raw symbols not included in $\{\yv^{(k)}_1\}_{k=1}^{K}$ will be delivered. Since we always keep the DoF allocation the same under both schemes, under the newly constructed updating scheme, the intended updates will be decoded no later than that under the original scheme, rendering an equivalent or even better AoI performance. 
\end{proof}

\begin{Definition}[Set of efficient updating schemes $\Pi_0$]\label{defn:Pi_0}
For the $(K,M,N,B)$ system with any given initial state, denote $\Pi_0$ as a set of deterministic updating schemes that deliver raw packets to users only while satisfying the following properties:
\begin{itemize}
    \item[i)] All transmitted updates will be decoded at the intended user and reset the corresponding AoI.
    \item[ii)] Any delivered update $\wv^{(k)}_t$ is transmitted starting from its generation time $t$.
    \item[iii)] The transmitter will utilize the maximum DoF during the transmission of any update unless in the time slot when the update is delivered. 
    \item[iv)] Among the symbols delivered to the same user, symbols generated earlier are delivered no later than symbols generated later. 
\end{itemize}
\end{Definition}

\begin{Theorem}\label{thm:deterministic_policy}
The updating scheme that achieves the minimum summed long-term average AoI lies in $\Pi_0$.
\end{Theorem}

\begin{proof}
First, we note that due to the deterministic system model, the optimal policy should be deterministic, as we can always execute the sample path that renders the minimum summed long-term average AoI under any randomized policy to outperform the original randomized policy.

Next, due to the deterministic setting, the system can foresee the AoI evolution under any deterministic updating scheme; Thus, it is unnecessary to transmit symbols that will not help to improve the AoI. 

Property ii) can be shown by noticing that starting to transmit an older update instead of the newly generated one at time $t$ leads to higher AoI when the update is delivered.

Property iii) is based on the following observation: Assume the DoF is not fully utilized during the transmission of an update $\wv^{(k)}_t$ before it is delivered in time $d$, i.e., in a time slot $t'$, $t\leq t'<d$, the DoF allocated to user $k$ is less than $N$, and the total DoF allocated to all users is less than $M$. Then we can allocate at least one more DoF to the user $k$ in time slot $t'$ without affecting the DoF allocation to other users. Thus, one more symbol from $\wv^{(k)}_t$ can be delivered in time $t'$, potentially reducing the time used to deliver $\wv^{(k)}_t$. Since an earlier delivery will strictly improve the AoI, the new policy performs better or at least the same as the original policy.

Property iv) can be proved through contradiction: assume under the optimal policy two updates $\wv^{(k)}_{t_1}$ and $\wv^{(k)}_{t_2}$, $t_1<t_2$, are delivered to user $k$ at time $d_1$ and $d_2$, $d_1<d_2$, respectively. Assume $d_1> t_2$. Since $\wv^{(k)}_{t_2}$ must be transmitted at time $t_2$, at least one of its symbols is delivered to user $k$ at time $t_2$. Meanwhile, since $\wv^{(k)}_{t_1}$ is not delivered until $d_1$, we can always switch the transmission of one symbol from $\wv^{(k)}_{t_2}$ with another symbol from $\wv^{(k)}_{t_1}$ that is delivered at $d_1$ under the original scheme. This potentially shortens the delivery time for $\wv^{(k)}_{t_1}$ without affecting the delivery time of $\wv^{(k)}_{t_2}$, which improves the AoI. 
\end{proof}

In the following, we will restrict to updating schemes in $\Pi_0$ only. Instead of considering the long-term average AoI, in the remaining of this section, we partition the time-axis into frames of length $B$, and investigate the minimum summed AoI in any frame. Since the summed long-term average AoI must be greater than the minimum time-average summed AoI in any frame, the latter serves as a lower bound for the former.

\begin{Lemma}\label{lemma:N=1}
For the $(K,M,1,B)$ system, under any policy $\pi\in\Pi_0$, during any consecutive $B$ time slots, at most $\min\{K,M\}$ updates are delivered, each to a different user.
\end{Lemma}
\begin{proof}
First, we note that the maximum DoF for the $(K,M,1,B)$ system in any time slot is $\min\{K,M\}$. Thus, the maximum number of linearly independent equations delivered in each frame is $B\min\{K,M\}$, which implies that at most $\min\{K,M\}$ updates can be decoded in any frame.
Next, we note that the maximum DoF for each user is 1 since it only has one receiving antenna. Thus, at most one update can be decoded for each user in any frame. 
Therefore, in any time frame, at most $\min\{K,M\}$ updates are delivered, each for a different user.
\end{proof}

\begin{Theorem}\label{coro:thm1-case1}
For the $(K,M,1,B)$ system with $K\leq M$, the summed AoI in frame consisting of time slots $[mB+1:(m+1)B]$, $m\in\Zb_{\geq 2}$, is lower bounded by $\frac{1}{2}K(3B-1)$.
\end{Theorem}

\begin{proof}
We consider the updating scheme that minimizes the summed AoI in the given time frame and ignore the AoI evolution outside the time window. The summed AoI in the frame is determined by the last update before time $mB+1$ and the update within the frame for each user. Denote the last updating time for the $K$ users prior to time $mB+1$ as $t_1\leq t_2\leq \ldots t_K\leq mB$. Then, we have the following observations.

First, if $t_1<(m-1)B+1$, we can always construct an alternative updating scheme under which another update is delivered to the same user at time $t_1+B$ without violating Lemma~\ref{lemma:N=1} and reduce its summed AoI in the frame considered. Thus, to obtain a lower bound on the summed AoI in the frame, we restrict to the scenario $t_1\geq (m-1)B+1$.

Next, we note that the summed AoI in the frame is minimum when the reset AoI at $t_1,t_2,\ldots,t_K$ are equal to $B$, as each update takes at least $B$ time slots to deliver.

Finally, we point out that the summed AoI in the frame can be minimized if the next updating happens exactly $B$ time slots after the previous updating for each user, i.e., at time $t_1+B$, $t_2+B$, $\ldots$, $t_K+B$.  

Calculating the cumulative AoI of each user during frame $m$, we have
\begin{align}
\delta_{mB+1:(m+1)B}^{(k)} &\geq \sum_{\ell=(m-1)B+1}^{t_k+B-1}(\ell-t_k+B)+\sum_{\ell=t_k+B}^{(m+1)B}(\ell-t_k) \nonumber \\
&=\sum_{\ell=B}^{2B-1}\ell
=\frac{1}{2}B(3B-1),  \label{eqn:thm1-case1:1}
\end{align}
and the summed AoI in the frame is lower bounded by $\frac{1}{2}KB(3B-1)$.
\end{proof}

\begin{Theorem}\label{coro:thm1-case2}
For the $(K,M,1,B)$ system with $K=pM+q$, where $p\in\Nb$, $q\in[0:M-1]$, the summed AoI in the frame consisting of time slots $[mB+1:(m+1)B]$, $m\in\Zb_{\geq p+1}$, is lower bounded by $\frac{1}{2}pM(pM+2B-1)+\frac{1}{2}q(2pB+3B-1)$.
\end{Theorem}

\begin{proof}
Similar to the $K\leq M$ case, the summed AoI in the frame is determined by the last update before time $mB+1$ and the update in the frame for each user. Denote the last updating time for the $K$ users prior to time $mB+1$ as $t_1\leq t_2\leq \ldots t_K\leq mB$. Then, we have the following observations.

Since $K>M$, according to Lemma~\ref{lemma:N=1}, at most $M$ users can be updated in each frame. Then, to obtain a lower bound on the summed AoI in the frame, we assume $t_{K-\ell M+1},\ldots,t_{K-(\ell-1)M}$ lie in the frame starting at $(m-\ell)B+1$, $1\leq \ell \leq p$, and $t_1,\ldots t_{q}$ lie in the frame starting at $(m-p-1)B+1$. This is because if the update times are not in the corresponding frames, we can always reschedule the transmission of updates without violating Lemma~\ref{lemma:N=1} and reduce the corresponding AoI contribution from those updates in the frame starting at $mB+1$.

Then, following the same argument as for the $K\leq M$ case, 
the summed AoI in the frame starting at $mB+1$ is minimum when the reset AoI at $t_1, t_2,\ldots, t_K$ are equal to $B$.
 
Besides, to minimize the summed AoI in the frame, the transmitter should update the $M$ users with the highest AoIs during the frame. Due to the constraints imposed by Lemma~\ref{lemma:N=1}, the updates should happen at time $t_{K-M+1}+B$, $t_{K-M+2}+B$, $\ldots$, $t_K+B$.  

Calculating the summed AoI of all users during the frame, we have the lower bound hold.
\end{proof}

\section{Converse of Theorem~\ref{result2_new}}\label{sec:converse-thm2}
In the following, we let $\bar{\delta}_t$ be the vector consisting of $\{\delta_t^{(k)}\}_k$ arranged in the increasing order, and name it the {\it AoI pattern} at time $t$. We note that the summed AoI in any time slot can be determined by the AoI pattern without considering the specific AoI at each user. We name the AoI pattern that renders the minimum summed AoI in any time slot as the {\it minimum AoI pattern}. We note that the summed long-term AoI is lower bounded by the sum of the AoIs in the minimum AoI pattern. 

For the first two cases in Theorem~\ref{result2_new}, we can obtain lower bounds as follows. 

For the case when $N\geq B$ and $\frac{M}{B}\geq 2$, the AoI at each user is lower bounded by one due to the transmission delay, i.e., the minimum AoI pattern at any time slot is $(1,1)$. Therefore, the summed average AoI is lower bounded by $2$. 

For the case when $N\geq B$ and $1\leq \frac{M}{B}<2$, at most one update generated at the beginning of time slot $t$ can be delivered. Thus, at the end of time slot $t$, at most one user can be updated with AoI reset as 1, while the other user is either not updated, or updated with AoI reset as 2. The minimum AoI pattern is thus $(1,2)$ and the summed AoI is lower bounded by 3. Those two lower bounds match with the AoI obtained under the updating schemes described in Section~\ref{sec:achievale-thm2-case1} and Section~\ref{sec:achievale-thm2-case2}, indicating the optimality of the updating schemes. 

In the following, we provide a matching lower bound for case (iii) in Theorem~\ref{result2_new}, i.e., when $N\geq M$ and $\frac{j}{ij+1}\leq \frac{M}{B}<\frac{j+1}{(j+1)i+1}$, for some $i,j\in\Nb$. The lower bound for Theorem~\ref{result2_new} (iv) can be derived similarly and deferred to Appendix~\ref{app:result2_iv}.

For a $(2,M,N,B)$ system with $N\geq M$, we note that the maximum DoF of the whole system is $M$, while the maximum DoF allocated to individual users is also $M$. The transmitter needs to decide how to split its DoF between the two users in each time slot. 

In the following, we first obtain a lower bound for a subset of policies named as alternating updating schemes, and then show that the lower bound applies to any policy lying in $\Pi_0$.

\begin{Definition}[Set of alternating updating schemes $\Pi_1$]\label{defn:alternating_updating}
Under an alternating updating scheme $\pi\in\Pi_1\subset\Pi_0$, in each time slot, the transmitter utilizes all of its DoF on a single user unless an update is decoded. Besides, the two users are updated alternately. 
\end{Definition}

\begin{remark}
We note that under the alternating policy, the user to be updated next is always the user with higher AoI.
\end{remark}

For any policy $\pi \in \Pi_1$, it can be represented as a sequence of blocks, where each block $B_{u,v}$ consists of $v$ idling time slots followed by $l_u:=\lceil uB/M\rceil$ time slots, during which the transmitter exhausts its DoF to send $u$  updates to the two users alternately. When $v=0$, we simply express $B_{u,0}$ as $B_u$. An updating scheme for the $(2,2,3,3)$ system is shown in Fig.~\ref{fig:thm2-block}, which can be represented by $(B_{3},B_{2,1},B_{1},\cdots)$ as illustrated.

\begin{figure}[t]
	\centering
	\includegraphics[width=3.4in]{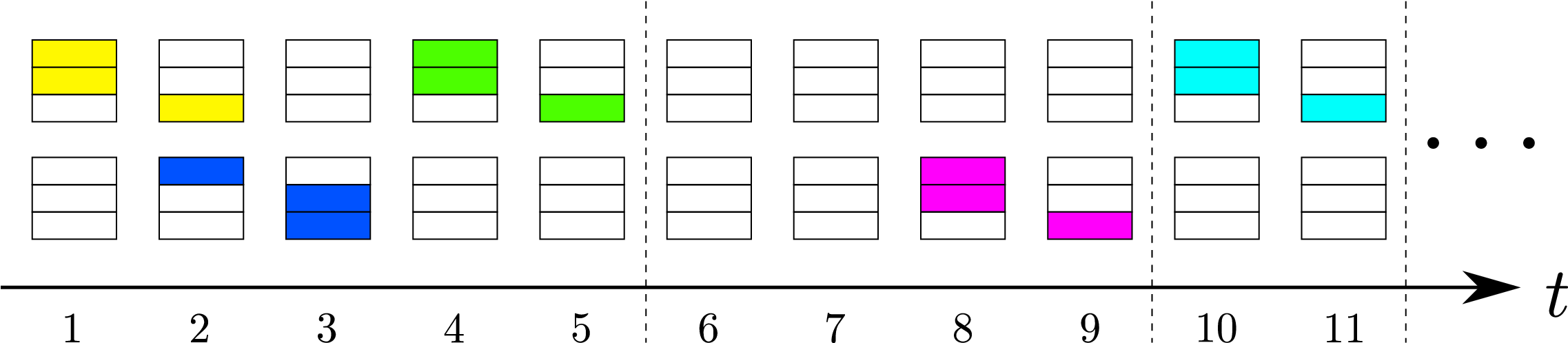}
	\captionof{figure}{An alternating updating scheme for the $(2,2,3,3)$ system, which can be equivalently represented as $(B_{3},B_{2,1},B_{1},\cdots)$.}
	\label{fig:thm2-block}
\end{figure}

\begin{Lemma}\label{lemma:thm2-case5:minAoI}
For the {$(2,M,N,B)$ system with $N\geq M$ and $\frac{j}{ij+1}\leq \frac{M}{B}<\frac{j+1}{(j+1)i+1}$}, $i,j\in\Nb$, the minimum AoI pattern in any time slot is $(i+1,2i+\lceil \frac{2}{j}\rceil)$.
\end{Lemma}
\begin{proof}
In order to decode an update, the transmitter needs to deliver at least $B$ linearly independent equations to the user. Due to the DoF constraint, it requires at least $\lceil B/M\rceil$ time slots. Since $i+\frac{1}{j} \geq \frac{B}{M}>i+\frac{1}{j+1}$, $\lceil \frac{B}{M}\rceil=i+1$ for any $i,j\in\Nb$. Thus, the minimum AoI for any user in any time slot is $i+1$. 

In order to update both users, it requires to deliver at least $2B$ linearly independent equations, which needs $2i+\lceil\frac{2}{j}\rceil$ time slots. Thus, the minimum AoI pattern is $(i+1,2i+\lceil\frac{2}{j}\rceil)$. 
\end{proof}
 
\begin{remark}\label{remark:min}
According to Lemma~\ref{lemma:thm2-case5:minAoI}, we can see that if $j=1$, the minimum AoI pattern is $(i+1,2i+2)$; if $j\geq 2$, the minimum AoI pattern becomes $(i+1,2i+1)$.
\end{remark}

In the following, we will first study a work-conserving updating scheme $\pi_1\in\Pi_1$, under which the transmitter exhausts its DoF at each time slot and update the two users continuously. By establishing the relationship between policy $\pi_1$ and block $B_u$, we will identify a lower bound on the summed average AoI over a block $B_u$, based on which we are able to obtain a lower bound on the summed average AoI for any policy in $\Pi_1$.

Without loss of generality, under policy $\pi_1$, we assume the initial AoI pattern at time 0 is the minimum AoI pattern $(i+1,2i+\lceil\frac{2}{j}\rceil)$.

\begin{Lemma}\label{lemma:thm2-duration}
For the {$(2,M,N,B)$ system with $N\geq M$ and $\frac{j}{ij+1}\leq \frac{M}{B}<\frac{j+1}{(j+1)i+1}$}, $i,j\in\Nb$, under policy $\pi_1$, the duration between two consecutive delivered updates is either $i$ or $i+1$ time slots. 
\end{Lemma}
\begin{proof}
Assume under policy $\pi_1$, the $m$-th update is delivered at time slot $t_m$. We will show that the $(m+1)$-th update is delivered either at time $t_m+i$ or at time $t_m+i+1$.

1) If the $m$-th update takes up all DoFs in time slot $t_m$, then, the $(m+1)$-th update is generated at time slot $t_{m}+1$, which must be delivered at time $t_m+i+1$ as it takes exactly $i+1$ time slots to deliver.

2) If the $m$-th update takes $D_m$ DoFs of time slot $t_m$ where $ D_m\in[1:M-1]$, then, the $(m+1)$-th update will be generated at time $t_m$ under policy $\pi_1$ and take up the remaining $M-D_m$ DoFs. It will be delivered at time $t_m+\lceil \frac{B-M+D_m}{M}\rceil$. Since $\lceil \frac{B}{M}\rceil=i+1$, we have
\begin{align}
&\left\lceil \frac{B-M+D_m}{M}\right\rceil\geq \left\lceil \frac{B+1}{M}\right\rceil-1\geq \left\lceil \frac{B}{M}\right\rceil-1=i, \\
&\left\lceil \frac{B-M+D_m}{M}\right\rceil \leq \left\lceil \frac{B-1}{M} \right\rceil \leq \left\lceil \frac{B}{M} \right\rceil =i+1.
\end{align}
Hence, it will be delivered either at time $t_m+i$ or at time $t_m+i+1$.
\end{proof}

Label the delivered updates starting at time $1$ in the order of their delivery time. 
Let $U_m$ be the index of the $m$-th update whose delivery time is $i+1$ time slots after the previous delivered update. 
Since the first update generated at time slot $1$ is delivered at the end of time slot $i+1$, we have $U_1=1$. We can see that the delivery times of the following $j-1$ delivered updates are exactly $i$ time slots after the previous delivery time while the $(j+1)$-th updating time is $i+1$ time slots after the previous update, hence $U_2=j+1$.
In general, for $m\in\Nb$, the duration between the delivery times of updates $U_m$ and $U_m-1$ equals $i+1$ and the durations between any other two consecutive updates are $i$. Thus, the delivery time for $U_m-1$ is at the end of time slot $(U_m-1)i+m-1$. 

By the DoF constraint, we have
\begin{align}
[i(U_m-1)+m-1]{M}\geq (U_m-1){B}, \label{eqn:thm2-case5:2-1}
\end{align}
i.e., the maximum DoF over $[1:(U_m-1)i+m-1]$ must be greater than the DoF required to deliver $U_m-1$ updates. 

Similarly, update $U_m$ is delivered at time $U_m i+m$. Thus,
\begin{align}
[i U_m+m-1]{M}< U_m B, \label{eqn:thm2-case5:2-2}
\end{align}
i.e., the maximum DoF over $[1:U_m i+m-1]$ must be less than the DoF required to deliver $U_m$ updates.

Eqn.~(\ref{eqn:thm2-case5:2-1}) and Eqn.~(\ref{eqn:thm2-case5:2-2}) imply that
\begin{align}
\frac{(m-1)M/B}{1-iM/B}<U_m\leq \frac{(m-1)M/B}{1-iM/B}+1. \label{eqn:thm2-case5:3}
\end{align}
Since $\frac{j}{ij+1}\leq \frac{M}{B}<\frac{j+1}{(j+1)i+1}$, we have
\begin{align}
j\hspace{-0.03in}-\hspace{-0.03in}1\hspace{-0.03in}<\hspace{-0.03in}\frac{M/B}{1\hspace{-0.03in}-\hspace{-0.03in}iM/B}\hspace{-0.03in}-\hspace{-0.03in}1\hspace{-0.03in}<\hspace{-0.03in}U_{m+1}\hspace{-0.03in}-\hspace{-0.03in}U_m\hspace{-0.03in}<\hspace{-0.03in}\frac{M/B}{1\hspace{-0.03in}-\hspace{-0.03in}iM/B}\hspace{-0.03in}+\hspace{-0.03in}1\hspace{-0.03in}<\hspace{-0.03in}j\hspace{-0.03in}+\hspace{-0.03in}2. \label{eqn:thm2-case5:4}
\end{align}
Under the constraint that $U_m$ and $U_{m+1}$ are integers, we must have
\begin{align}
U_{m+1}-U_m=j\quad \mbox{or} \quad j+1.   \label{eqn:thm2-case5:5}
\end{align}

We now partition the time axis into segments by the delivery time of updates $\{U_m-1\}_{m=2}^{\infty}$, i.e., $[1:ij+1]$, $\ldots$, $[(U_m-1)i+m:(U_{m+1}-1)i+m]$, $\ldots$.
According to Eqn.~(\ref{eqn:thm2-case5:5}), the segment length is either $ij+1$ or $(j+1)i+1$. An example of the definition of $U_m$ and the segments for the $(2,7,12,12)$ system is illustrated in Fig.~\ref{fig:segment_ex}.

\begin{figure}[t]
	\centering
	\includegraphics[width=3.4in]{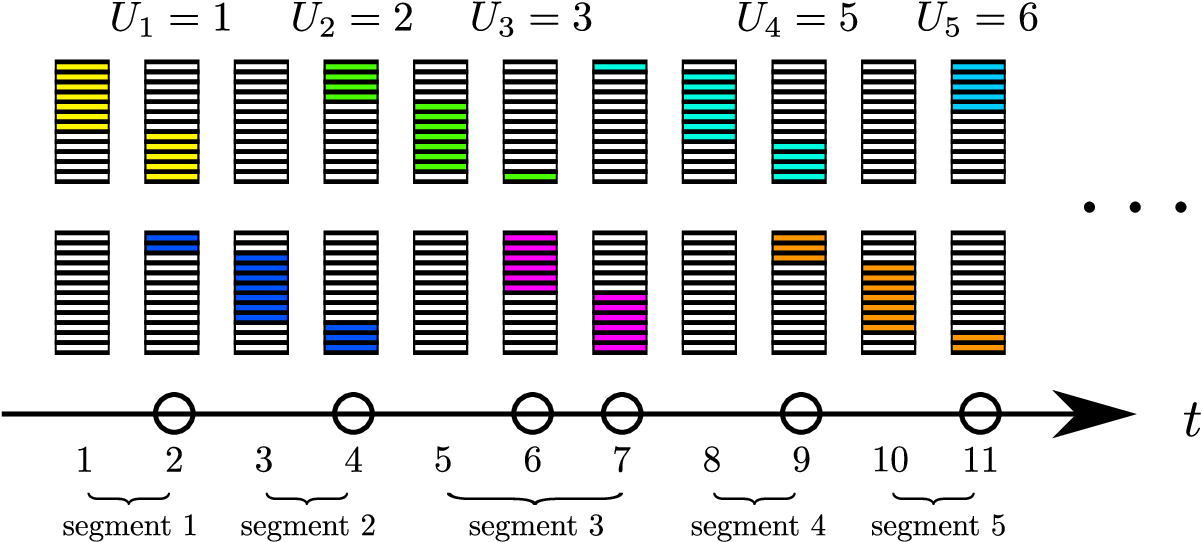}
	\captionof{figure}{Updating pattern in the $(2,7,12,12)$ system under $\pi_1$. We have $i=1$, $j=1$. Circles represent delivery times of updates. Since $i=1$, the segments can be obtained by tracking the updates whose delivery time is 2 time slots after the previous delivery time. We note that the length of each segment is either 2 (i.e., $ij+1$) or 3 (i.e., $(j+1)i+1$). }
	\label{fig:segment_ex}
\end{figure}

\begin{table*}
\caption{Minimum AoI pattern for segments of length $ij+1$, $j\geq 2$, $\ell\in[2:j-1]$. An update is delivered at the end of the time slots in the last column.}
\centering
{\small
\begin{tabular}{|c|c|c|c|c|}\hline
Time slot  & $(U_m-1)i+m$   & $\cdots$ & $U_m i+m-1$ & $U_m i+m$     \\ \hline
Minimum AoI pattern & $(i+2,2i+2)$ & $\cdots$ & $(2i+1,3i+1)$ & $(i+1,2i+2)$ \\  \hline\hline
Time slot   & $U_m i+m+1$  & $\cdots$ & $(U_m+1)i+m-1$   & $(U_m+1)i+m$    \\ \hline
Minimum AoI pattern  & $(i+2,2i+3)$ & $\cdots$ & $(2i,3i+1)$ & $(i+1,2i+1)$  \\  \hline\hline
Time slot    & $(U_m+\ell-1)i+m+1$  & $\cdots$ & $(U_m+\ell)i+m-1$   &  $(U_{m}+\ell)i+m$   \\ \hline
Minimum AoI pattern & $(i+2,2i+2)$ & $\cdots$ & $(2i,3i)$ & $(i+1,2i+1)$  \\ \hline
\end{tabular}
}
\label{tab:thm2:1}
\end{table*}

\begin{table*}
\caption{Minimum AoI pattern for segments of length $(j+1)i+1$, $j\geq 2$, $\ell\in[3:j]$. An update is delivered at the end of the time slots in the last column.}
\centering
{\small
\begin{tabular}{|c|c|c|c|c|}\hline
Time slot  & $(U_m-1)i+m$    & $\cdots$ & $U_m i+m-1$  & $U_m i+m$    \\ \hline
Minimum AoI pattern & $(i+2,2i+2)$  & $\cdots$ & $(2i+1,3i+1)$ & $(i+2,2i+2)$ \\  \hline\hline
Time slot   & $U_m i+m+1$  & $\cdots$ & $(U_m+1)i+m-1$   & $(U_m+1)i+m$    \\ \hline
Minimum AoI pattern  & $(i+3,2i+3)$ & $\cdots$ & $(2i+1,3i+1)$ & $(i+1,2i+2)$ \\  \hline\hline
Time slot    & $(U_m+1)i+m+1$  & $\cdots$ & $(U_m+2)i+m-1$   & $(U_m+2)i+m$   \\ \hline
Minimum AoI pattern  & $(i+2,2i+3)$ & $\cdots$ & $(2i,3i+1)$  & $(i+1,2i+1)$  \\  \hline\hline
Time slot    & $(U_m+\ell-1)i+m+1$  & $\cdots$ & $(U_m+\ell)i+m-1$ & $(U_{m}+\ell)i+m$       \\ \hline
Minimum AoI pattern  & $(i+2,2i+2)$ & $\cdots$ & $(2i,3i)$ & $(i+1,2i+1)$  \\  \hline
\end{tabular}
}
\label{tab:thm2:2}
\end{table*}

\begin{Lemma}\label{lemma:thm2-case3:seg1}
For the {$(2,M,N,B)$ system with $N\geq B$ and $\frac{j}{ij+1}\leq \frac{M}{B}<\frac{j+1}{(j+1)i+1}$}, $i,j\in\Nb$, under policy $\pi_1$, the summed average AoI over segment $[(U_m-1)i+m:(U_{m+1}-1)i+m]$ is lower bounded by $\Delta_{\min}=4i+1+\frac{2i+1}{ij+1}$ if $j\in\Zb_{\geq 2}$, and by $\Delta_{\min}=4i+3$ if $j=1$.
\end{Lemma}

\begin{proof}
{\it 1) $j\geq 2$.} 
We start with the case when the segment length is $ij+1$, i.e., $U_{m+1}-U_m=j$. According to Remark~\ref{remark:min}, the minimum AoI pattern at the end of time slot $(U_m-1)i+m-1$ is $(i+1, 2i+1)$. Hence, the minimum AoI pattern at the first time slot of the segment starting at $(U_m-1)i+m$ is $(i+2,2i+2)$. We note that under $\pi_1$, update $U_m$ is delivered at time $U_m i+m$, with minimum age $i+1$. This would happen if $U_m$ is generated at time $(U_m-1)i+m$. After that, updates $U_m+1,\ldots,U_m+j-1$ are delivered sequentially after $i$ time slots since the previous delivery. Thus, the minimum age of those updates when delivered is $i$. Due to the alternating updating structure, the user with higher AoI will always be updated next under $\pi_1$. Thus,  the minimum AoI pattern over the segment can thus be specified (cf. Table~\ref{tab:thm2:1}), and the corresponding minimum summed average AoI over the duration is $4i+1+\frac{2i+1}{ij+1}$.

Next, we consider the case when $U_{m+1}-U_m=j+1$ and the corresponding segment length is $(j+1)i+1$. 
We will show that the minimum AoI pattern when $U_m$ is delivered, i.e., at the end of time slot $U_m i+m$, is $(i+2,2i+2)$ instead of $(i+1,2i+2)$, i.e., update $U_m$ must be generated at time slot $(U_m-1)i+m-1$ instead of $(U_m-1)i+m$. We prove it by contradiction. 

Assume update $U_m$ is generated at time slot $(U_m-1)i+m$. Since update $U_m-1$ is delivered at time $(U_m-1)i+m-1$, update $U_m$ would consume all DoF at time slot $(U_m-1)i+m$ under policy $\pi_1$. Thus, under policy $\pi_1$, the DoF allocation for updates $U_m,\ldots,U_{m+1}-1$ would be the same as that for updates $U_1,\ldots,U_2-1$. Therefore, the length of segment $[(U_m-1)i+m:(U_{m+1}-1)i+m]$ would be identical to that of segment $[1,ij+1]$, i.e., $ij+1$. This contradicts with the assumption that the segment is of length $(j+1)i+1$, which indicates that update $U_m$ must be generated at time slot $(U_m-1)i+m-1$, and reset the AoI of the corresponding user to $i+2$ instead of $i+1$ when delivered. 

With the minimum AoI pattern at the end of time slot $U_m i+m$ by $(i+2,2i+2)$, the minimum AoI pattern can be identified (cf. Table~\ref{tab:thm2:2}). The minimum summed average AoI over the segment can thus be calculated, which is equal to $4i+1+\frac{4i+1}{(j+1)i+1}$.

Combining those two cases, we can see that the summed average AoI over any segment is lower bounded by $4i+1+\frac{2i+1}{ij+1}$.

{\it 2) $j=1$.} For this case, the segment length is either $i+1$ or $2i+1$. According to Remark~\ref{remark:min}, the minimum AoI pattern is $(i+1,2i+2)$. If the segment length is $i+1$, there is only one update at the end of the segment, which resets the AoI to $i+1$. The corresponding summed average AoI over the segment can be calculated, which is equal to $4i+3$.

When the segment length is equal to $2i+1$, two updates are delivered over the segment, one is at time $i$ and the other is at time $2i+1$. With the minimum AoI pattern $(i+1,2i+2)$, we can show that the summed average AoI is still lower bounded by $4i+3$.
\end{proof}

\begin{remark}\label{remark:decrease}
We note that for all $i,j\in\Nb$, the minimum summed average AoI over the first $\ell i+1$ time slots in each segment $[(U_m-1)i+m:(U_{m+1}-1)i+m]$ is monotonically decreasing in $\ell$ for $\ell\in[1:U_{m+1}-U_m]$. 
\end{remark}

Next, we relate the AoI pattern under $\pi_1$ with block $B_u$ under any alternative updating policy in $\Pi_1$. Recall that $l_u$ is the number of time slots required to deliver $u$ updates in a block $B_u$. Then, the updating scheme $\pi_1$ over $[1,l_u]$ is identical to a block $B_u$ except that some DoF at time slot $l_u$ under $B_u$ may not be exhausted.

We note that $B_u$ can be partitioned into segments $[1:ij+1]$, $\ldots$, $[(U_{m_u-1}-1)i+m_u-1:(U_{m_u}-1)i+m_u-1]$ and a residue $[(U_{m_u}-1)i+m_u:l_u]$, where $m_u=\max\{m:U_m< l_u\}$. According to Remark~\ref{remark:decrease}, the summed average AoI of the residue is lower bound by $\Delta_{\min}$.

\begin{Lemma}\label{lemma:ineq}
If $x,y,z,w,t\in\Rb_{>0}$ satisfy inequalities $\frac{x}{y}\geq t$ and $\frac{z}{w}\geq t$, then $\frac{x+z}{y+w}\geq t$.
\end{Lemma}

Since the summed average AoI over each segment is lower bounded by the quantity in Lemma~\ref{lemma:thm2-case3:seg1}, then, based on Lemma~\ref{lemma:ineq}, the summed average AoI over any block $B_u$ is lower bounded by the quantity as well.

Next, we will show that the lower bound for block $B_u$ is also a valid lower bound for blocks $B_{u,v}$, $\forall v>0$.

\begin{Lemma}\label{lemma:idling}
For the {$(2,M,N,B)$ system with $N\geq M$ and $\frac{j}{ij+1}\leq \frac{M}{B}<\frac{j+1}{(j+1)i+1}$}, $i,j\in\Nb$, the summed average AoI over $B_{u,v}$ is lower bounded by $\Delta_{\min}=4i+1+\frac{2i+1}{ij+1}$ if $j\in\Zb_{\geq 2}$, and by $\Delta_{\min}=4i+3$ if $j=1$. 
\end{Lemma}
\begin{proof}
Recall that in a block $B_{u,v}$, there are $v$ idle time slots before $B_u$. Let $\delta_0^{(1)}$, $\delta_0^{(2)}$ be the AoI at time zero. Without loss of generality, we assume $B_{u,v}$ starts at time slot $1$. Since there is no updating over the first $v$ time slots, the AoI of user $k$, $k=1,2$, will monotonically increase until the first successful update at time $v+t_k$. Thus, the existence of idling time slots affects the AoI evolution until the first update for each user occurs. Let $v+l_u$ be the end of block $B_{u,v}$.

Let $A_1^{(k)}$ be the AoI increment induced by the idling time slots at user $k$, as illustrated by the shaded area in Fig.~\ref{fig:idling}. Meanwhile, denote $A_2^{(k)}$ as the summed AoI over $B_u$ when no idling time slot is present, corresponding to the unshaded area in Fig.~\ref{fig:idling}.

Then, we have
\begin{align}
A_1^{(k)}=\sum_{\ell=1}^{v}(\delta_0^{(k)}+\ell)+v\cdot t_k, \quad k=1,2.
\end{align}
Let $\Delta_{B_{u,v}}$ be the summed average AoI over $B_{u,v}$. Then, $\Delta_{B_{u,v}}=\frac{\sum_{k=1}^{2}(A_1^{(k)}+A_2^{(k)})}{v+l_u}$. We note that
\begin{align}
\frac{\sum_{k=1}^{2}A_1^{(k)}}{v}&=\delta_0^{(1)}+\delta_0^{(2)}+(1+v)+t_1+t_2  \nonumber \\
&\geq 2\left(3i+1+\left\lceil\frac{2}{j}\right\rceil\right) + (1+v),
\end{align}
where the last inequality follows from Lemma~\ref{lemma:thm2-case5:minAoI}. 

Note that $6i+2\lceil \frac{2}{j}\rceil +v+1>4i+1+\frac{2i+1}{ij+1}$ if $j>2$ and $6i+2\lceil \frac{2}{j}\rceil +v+1>4i+3$ if $j=1$. By applying Lemma~\ref{lemma:ineq}, we have the lower bounds $\Delta_{\min}$ hold for $B_{u,v}$ as well.
\end{proof}

\begin{figure}[t]
	\centering
	\includegraphics[width=2.4in]{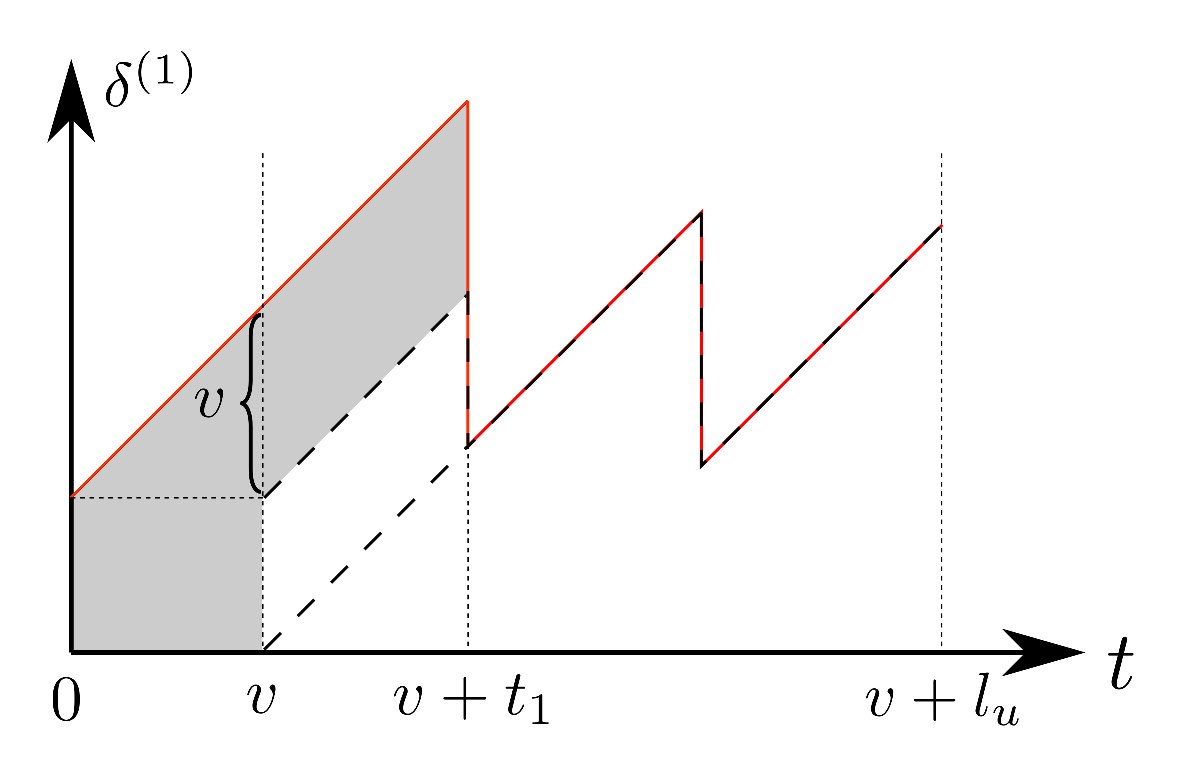}
	\captionof{figure}{AoI evolution of user 1 over an extended block. The dashed line is the AoI evolution if idling period does not exist.} 
	\label{fig:idling}
\end{figure}

Since every policy in $\Pi_1$ can be decomposed as a sequence of blocks in the form of $B_{u,v}$, the lower bound on each block applies to the long-term average. Thus, $\Delta_{\text{min}}$ is a lower bound on the summed average AoI for all policies in $\Pi_1$. To prove that the lower bound applies to all policies, it suffices to show that no other policy in $\Pi_0$ can achieve AoI lower than $\Delta_{\text{min}}$.

\begin{Theorem}\label{thm:deterministic_policy_1}
For the $(2,M,N,B)$ system with $N\geq M$ and $M/B\in [\frac{j}{ij+1},\frac{j+1}{(j+1)i+1})$, $i,j\in\Nb$, the summed average AoI under any policy $\pi\in\Pi_0$ is lower bounded by $\Delta_{\min}=4i+1+\frac{2i+1}{ij+1}$ if $j\in\Zb_{\geq 2}$, and by $\Delta_{\min}=4i+3$ if $j=1$. 
\end{Theorem}

The proof of Theorem~\ref{thm:deterministic_policy_1} is provided in Appendix~\ref{app:1}.

\section{Proof of Theorem~\ref{result3_new}}\label{sec:discuss}
For the two-user system with $N<B$ and $N<M<2N$, it is extremely challenging to identify the exact minimum average AoI due to the combinatorial nature of the problem. In this section, we strive to obtain a lower bound on the summed average AoI and propose a $2$-optimal policy.

\subsection{Lower Bound}\label{sec:discuss:lb}
We provide a lower bound on the summed average AoI, as summarized in the following lemma.
\begin{Lemma}\label{lemma:ineq-minAoI}
For the $(2,M,N,B)$ system with $N<M<2N$, $B=iN+j$, $i\in\Nb$, $j\in[0:N-1]$, the summed average AoI is lower bounded by $\Delta_{\mathsf{LB}}:= 2i+\lceil\frac{2j}{N}\rceil$.
\end{Lemma}
\begin{proof}
Denote $d_1$ and $d_2$ are the AoI at user 1 and user 2 at an time slot $k$, respectively. Without loss of generality, assume $d_1\leq d_2$. Let $k-d_1$ and $k-d_2$ as the generation times of the freshest update received at user 1 and user 2, respectively, and $u_1$ and $u_2$ as the corresponding updates. Define $x_\ell$ and $y_\ell$ as the DoFs allocated for the transmission of updates $u_1$ and $u_2$ at time slot $k-\ell$, respectively. Then, we must have the following conditions satisfied:
\begin{eqnarray}
& &\sum_{\ell=1}^{d_1} x_\ell =B, \quad \sum_{\ell=1}^{d_2} y_\ell =B, \label{eqn:++:4} \\
& &0\leq x_\ell+y_\ell\leq M,\quad  \ell=1,\ldots, d_2, \label{eqn:++:5} \\
& &0\leq x_\ell, y_\ell\leq N,\quad  \ell=1,\ldots, d_2.  \label{eqn:++:6}
\end{eqnarray}

Thus,
\begin{align}
2B=\sum_{\ell=1}^{d_1}(x_\ell+y_\ell)+\sum_{\ell=d_1+1}^{d_2}y_\ell\leq d_1M + (d_2-d_1)N.
\end{align}
Since $B=iN+j$ and $M<2N$, the above inequality becomes
\begin{align}
d_1+d_2\geq 2i+\frac{2j}{N}+\frac{d_1(2N-M)}{N}>2i+\frac{2j}{N}.
\end{align}
Besides, since $d_1+d_2$ is an integer, we must have $d_1+d_2\geq 2i+\lceil\frac{2j}{N}\rceil$, which provides a valid lower bound on the summed average AoI at any time slot $k$.
\end{proof}

\subsection{Framed Alternating Updating}\label{sec:discuss:scheme}
We propose a framed alternating updating scheme and provide its performance guarantee subsequently.
\begin{Definition}[Framed alternating updating]
Partition the time axis into frames of length $2i$ if $j=0$, $2i+1$ if $0<2j<M$, or $2i+2$ otherwise.
\begin{itemize}
\item[1)] If $j=0$, or $2j>M$, within each frame, the transmitter first utilizes its transmission capability to update the user with higher AoI until its AoI resets. Then, at the beginning of the next time slot, the transmitter starts to transmit a new update to the other user until the end of the frame.
\item[2)] If $0<2j<M$, in each frame, the transmitter first utilizes its transmission capability to update the user with higher AoI until its AoI resets at the $i$-th time slot. Within the same time slot, the transmitter exhausts the remaining transmission capability to start transmitting an new update for the other user until the end of the frame.
\end{itemize}
\end{Definition}
\begin{remark}
For the $(2,M,N,B)$ system with {$N<B$} and $N<M<2N$, in each time slot, it is important to decide whether or not to exploit the remaining $M-N$ DoFs when $N$ DoFs have been used to update one user. Exploiting the remaining DoFs may lead to earlier updating of the other user, while wasting them may shorten the transmission time of an update and reduces its age when delivered. When $M$ gets close to $N$, the benefit of utilizing the $M-N$ remaining DoFs is offset by the elongated age of the update. As a result, we expect that the framed alternating updating scheme performs close to optimal when $M$ approaches $N$.
\end{remark}

In order to characterize the AoI performance, we track the AoI evolution under the framed alternating updating scheme. Note that when $j=0$, each user is updated every $2i$ time slots, and when it is updated, its AoI is reset as $i$ and starts increasing until next update. The summed average AoI thus equals $\Delta=4i-1$. When $2j>M$, similar analysis shows that the summed average AoI $\Delta=4i+3$. When $0<2j<M$, each user is updated every $2i+1$ time slots, and when it is updated, its AoI is reset as $i+1$. Therefore, the summed average AoI is $\Delta=4i+1$.

Note that the lower bound in Theorem~\ref{result3_new} becomes $\Delta_{\mathsf{LB}}=2i$ if $j=0$, $\Delta_{\mathsf{LB}}=2i+1$ if $0<2j<M$, and $\Delta_{\mathsf{LB}}=2i+2$. Combining the summed average AoI of the framed alternating updating scheme and the lower bound, we have
\begin{align}
\frac{\Delta}{\Delta_{\mathsf{LB}}}\leq \max\left\{\frac{4i-1}{2i},\frac{4i+1}{2i+1},\frac{4i+3}{2i+2}\right\}< 2,
\end{align}
i.e., in $(2,M,N,B)$ system with $N<B$ and $N<M<2N$, the summed average AoI under the framed alternating updating scheme is at most twice the minimum summed average AoI and the proposed policy is $2$-optimal.

\section{Conclusions and discussions}\label{sec:conclusion} 
In this paper, we investigated the AoI optimization problem in MIMO broadcast channels with various numbers of users, transmitting and receiving antennas and update sizes. Due to the combinatorial nature of the problem and the complex AoI evolution in a dynamic system, identifying the optimal updating scheme becomes challenging. 
We considered two specific scenarios, where in the first scenario, each receiver has one antenna, and in the second scenario, it only has two users. We developed different updating schemes for those cases and showed their optimality through rigorous analysis. Although the optimal schemes seem intuitive, establishing their optimality is non-trivial. Toward that, we developed some novel approaches. We think those approaches will be useful for the AoI-optimal updating schemes in noise-free MIMO broadcast channels with other parameters. Besides, we expect that those techniques can be extended to handle more practical noisy channels by leveraging the deterministic channel models proposed in \cite{Avestimehr:2011}. Due to the coupled dynamics of AoI evolution, the general AoI-tradeoff among multiple users are intractable. However, we expect that the approaches developed in this paper can be adopted to identify certain Pareto optimal points on the AoI of multiple users. We leave this as one of our future steps.

\appendix

\subsection{Proof of Theorem~\ref{thm:deterministic_policy_1}}\label{app:1}
First, we point it out that for the $(2,M,N,B)$ systems, under any policy in $\Pi_0$, at any time $t$, there exist at most two updates that are partially transmitted in any time slot. This is due to property iv) in Definition~\ref{defn:Pi_0}, i.e., the transmitter will not start transmitting a new update to a user until the previous one has been delivered to the same user. Therefore, at any time $t$, there exists at most one partially transmitted update for each user.

Next, for $(2,M,N,B)$ systems with $B>M$, we have the following observation.
\begin{Lemma}\label{lemma:prioritize_next_delivery}
For the $(2,M,N,B)$ system with $B>M$, consider two consecutive successful deliveries of updates from the transmitter under the optimal policy in $\Pi_0$. Denote their delivery times as $d_1,d_2$, respectively, $d_1\leq d_2$, and the corresponding generation times as $t_1$, $t_2$. With a little abuse of notation, we name those two updates $\wv_1$ and $\wv_2$, respectively. Then, either of the following two scenarios must be true: 1) $t_1< d_1\leq t_2< d_2$. 2) $t_2< t_1< d_1\leq d_2$, and over time $t_1\leq t< d_1$, the transmitter utilizes all of its DoF to transmit $\wv_1$. 
\end{Lemma}
\begin{proof}
Recall that for all policies in $\Pi_0$, the transmitter only uses its DoF to deliver updates that eventually reset the AoI. In the following, we show that any policy $\pi\in\Pi_0$ that violates the structures can be strictly improved to reduce the AoI. We consider the following cases: 

i) $t_1\leq t_2<d_1\leq d_2$. Recall that for any policy in $\Pi_0$, all delivered symbols are transmitted starting from their generation times. Thus, the transmitter must begin to transmit $\wv_2$ at time $t_2$. We then consider an alternative policy under which the transmitter will utilize the DoF that was allocated to transmit $\wv_2$ under the original policy at time slot $t_2$ and afterwards  for $\wv_1$ until $\wv_1$ is delivered. Apparently, $\wv_1$ will be delivered no later than $d_1$, which potentially improves the AoI of the corresponding user. After that, the transmitter will reallocate the DoF that was allocated for $\wv_1$ and $\wv_2$ to transmit a new update $\wv'_2$, where $\wv'_2$ and $\wv_2$ are intended for the same user. Since the total allocated DoF remains the same under both policies, it ensures that $\wv'_2$ will be delivered at $t_2$, which will reset the corresponding AoI with a smaller age. Thus, the overall AoI will be strictly improved under the alternative policy, indicating that this case cannot exist under the optimal policy.

ii) $t_2< t_1< d_1\leq d_2$, and there exists at least one time slot $t$, $t_1\leq t< d_1$, during which the transmitter does not exhaust its DoF to transmit $\wv_1$. Following the similar argument as in case i), we can construct an alternative policy under which the transmitter exhausts its DoF to transmit $\wv_1$ until its delivered, and utilizes the remaining DoF to deliver $\wv_2$. This will improve the AoI of the user that decodes $\wv_1$, without impacting the AoI of the other user. Thus, we can safely exclude this case for the optimal policies in $\Pi_0$ without compromising the optimality.
\end{proof}

Then, in order to show that the lower bound in Theorem~\ref{thm:deterministic_policy_1} applies to all policies in $\Pi_0$ for the {$(2,M,N,B)$} system with $B>M$, the optimal policy in $\Pi_0$ must exhibit the following structural properties.

\begin{Lemma}\label{lemma:update_higerAoI}
For the $(2,M,N,B)$ system with $B>M$, under any optimal policy in $\Pi_0$, a successful updating always updates the user with higher AoI.
\end{Lemma}
\begin{proof}
We consider an updating policy starting at a time slot $t_0$. 
Assume at the beginning of $t_0$, the AoI at users 1 and 2 are $\delta_0^{(1)}$, $\delta_0^{(2)}$, respectively, where $\delta_0^{(1)}>\delta_0^{(2)}$. 

Denote the first two delivered updates after time $t_0$ as $\wv_1$ and $\wv_2$. We assume their transmission pattern complies with Lemma~\ref{lemma:prioritize_next_delivery}. We aim to show that these two updates always update user 1 and then user 2. We consider the two possible transmission structures separately.

i) $t_2< t_1< d_1\leq d_2$. First, we note that those two updates are intended for different users. This is because if both $\wv_1$ and $\wv_2$ are intended for the same user, then the delivery of $\wv_2$ will not reset the AoI at the corresponding user, as $\wv_2$ is more stale than $\wv_1$. This violates the assumption that under the optimal policy in $\Pi_0$, all delivered updates reset the corresponding AoI. 

Next, we assume $\wv_1$ and $\wv_2$ are intended for user 2 and user 1, respectively. We note that under this scheme, user 2 will be updated at $d_1$ while user 1 is updated at $d_1$. We aim to show that this is strictly sub-optimal. For that, we consider an alternative policy where the transmitter replaces each update delivered after $t_0$ with an update generated at the same time but intended for the other user. We note that the AoI evolution at both users remains the same under both policies up to $d_1$, and are switched after $d_2$, as illustrated in Fig.~\ref{fig:adj_diff_user}. Between $d_1$ and $d_2$, since $\delta_0^{(1)}>\delta_0^{(2)}$, resetting user 1 at $d_1$ instead of $d_2$ leads to reduced summed AoI. Therefore, $\wv_1$ and $\wv_2$ should be intended for user 1 and user 2, respectively, under the optimal policy.

\begin{figure}[t]
	\centering
	\includegraphics[width=3.4in]{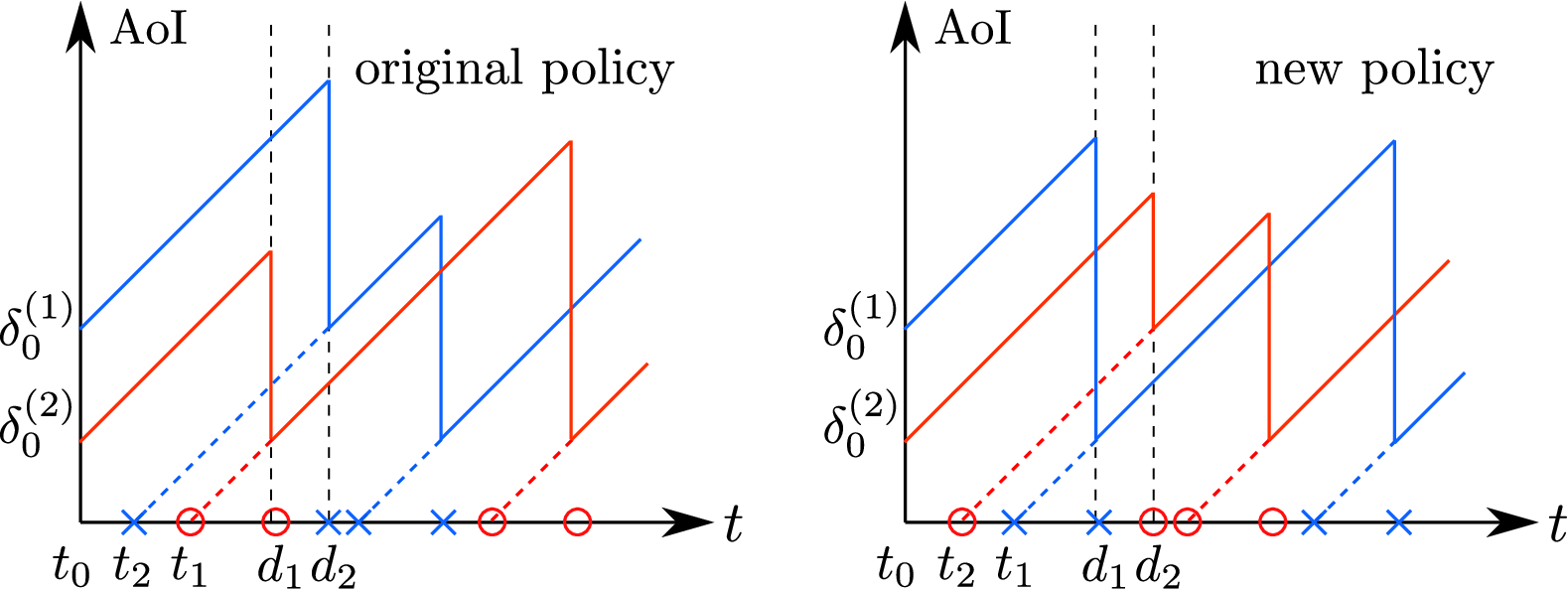}
	\captionof{figure}{Comparison between the new policy and the original policy.} 
	\label{fig:adj_diff_user}
\end{figure}

ii) $t_1< d_1\leq t_2< d_2$. We now consider the following cases: 

ii-a) $\wv_1$ and $\wv_2$ are both intended for user 2. For this case, we construct a new policy by replacing $\wv_1$ with another update $\wv_1'$ generated at the same time but intended for user 1. Then, under the new policy, the AoI of user 1 will be reset at time $d_1$, while the AoI of user 2 will be reset at $d_2$ only. Therefore, after $d_2$, the AoI of user 2 remains the same under both policies, while the AoI of user 1 will be strictly improved. Besides, between $d_1$ and $d_2$, the summed AoI of both users is strictly improved under the new policy, as resetting user with higher age (user 1) leads to lower summed AoI. Therefore,  $\wv_1$ and $\wv_2$ cannot be intended for user 2 under the optimal policy.

ii-b) $\wv_1$ and $\wv_2$ are intended for user 2 and user 1, respectively. For this case, we construct a new policy by replacing each update transmitted after $t_0$ by the update generated at the same time but intended for the other user. Then, under the new policy, after $d_2$, the AoI evolution of user 1 and user 2 will be switched. Between $d_1$ and $d_2$,  the summed AoI of both users is strictly improved under the new policy, as resetting user with higher age (user 1) leads to lower summed AoI.  

Combining cases ii-a) and ii-b), we can see that $\wv_1$ must be intended for user 1 under the optimal policy. 

Therefore, for the two possible updating structures, the next update must be intended for the user with higher AoI. Since $\wv_1$ and $\wv_2$ must intend for user $1$ and user $2$, respectively, for the first structure, we repeat this argument for updates after $d_2$. For the second structure, we only showed that $\wv_1$ must intend for user 1, while $\wv_2$ may intend for either user, depending on the updating structure after $d_1$. We then repeat the argument after $d_1$ for the second structure. Then, we can conclude that each delivered update should update the user with higher AoI. 
\end{proof}

\begin{remark}\label{reamrk:lemma-alternating_delivery-2}
Note that the delivery structure described in Lemma~\ref{lemma:update_higerAoI} does not necessarily imply the alternating transmission structure described in Definition~\ref{defn:alternating_updating}, due to the second possible transmission pattern depicted in Lemma~\ref{lemma:prioritize_next_delivery}.
\end{remark}

For ease of exposition, let $\tilde{\Pi}_0\subset \Pi_0$ be the set of policies satisfying Lemmas~\ref{lemma:prioritize_next_delivery}-\ref{lemma:update_higerAoI}.
According to Theorem~3, no free DoF is available during the transmission of an update, which naturally leads to the definition of generalized blocks as follows.

\begin{Definition}[Generalized Block $\tilde{B}_{u,v}$]
Block $\tilde{B}_{u,v}$ consists of $v$ idling time slots followed by $\lceil uB/M\rceil$ time slots during which the transmitter exhausts its DoF to send $u$ \emph{useful} updates to the two users. When $v=0$, we simply express $\tilde{B}_{u,0}$ as $\tilde{B}_u$.
\end{Definition}

Compared with the block $B_{u,v}$ defined in Section~\ref{sec:converse-thm2}, in generalized blocks $\tilde{B}_{u,v}$, we do not impose the alternating updating structure. Any policy $\pi\in\tilde{\Pi}_0$ can be represented as a sequence of generalized blocks. 

Next, we consider the DoF allocation within each generalized block. We introduce the definition of resource chunk as follows.

\begin{Definition}[Resource chunk] A resource chunk in block $\tilde{B}_{u,v}$ is the smallest subset of the utilized $uB$ DoFs in the block satisfying the following conditions: 1) At least one update is delivered using the DoF in each chuck; 2) During the transmission of the update(s) satisfying 1), there does not exist any other partially transmitted update in the system. 
\end{Definition}

\begin{Lemma}\label{lemma:chunk}
There are two types of resource chunks in each  block $\tilde{B}_{u,v}$: Type-1: A chunk consisting of $B$ DoFs allocated to the transmission of a single update. Type-2: A chunk consisting of $2B$ DoFs allocated to the transmission of two updates (denoted as $\wv_1$ and $\wv_2$) with $t_2< t_1< d_1\leq d_2$. Besides, the transmission time of $\wv_1$ is always equal to $i+1$.
\end{Lemma}

Lemma~\ref{lemma:chunk} can be shown based on the structure depicted in Lemma~\ref{lemma:prioritize_next_delivery}.

\begin{Definition}[Type-2 resource chunk re-allocation] A re-allocated Type-2 resource chunk will allocate the DoFs in the original resource chunk to users in the order of their original updating times. Moreover, all future updates delivered after this chunk are replaced by updates with the same generation time but intended for the other user.
\end{Definition}

\begin{remark}\label{reamrk:lemma-alternating_delivery-1-chunk}
Based on Lemma~\ref{lemma:update_higerAoI}, the first delivered update using a Type-2 resource chunk, i.e., the update delivered at time slot $d_1$, is always intended for the user with higher AoI. After re-allocation, such user should receive the first update as well.  

We note that when all Type-2 resource chunks are re-allocated, the corresponding updating schemes becomes an alternating updating policy in $\Pi_1$. Thus, to show that the summed average AoI under any policy in $\tilde{\Pi}_0$ is lower bounded by the same quantity $\Delta_{\min}$ suggested in Theorem~\ref{lemma:thm2-case3:seg1}, it suffices to show that the re-allocation of Type-2 resource chunks always improve the summed average AoI.
\end{remark}

As we have shown in the proof of Lemma~\ref{lemma:thm2-case5:minAoI}, we first note that the minimum possible transmission time for one update is $i+1$ time slots under any updating policy in $\Pi_0$.

\begin{Lemma}\label{lemma:thm2-chunk_i+1}
If after the re-allocation of a Type-2 resource chunk, the transmission time of the first delivered update is $i+1$, then the re-allocation always improves the AoI under the original resource chunk allocation.
\end{Lemma}
\begin{proof}
Let $\delta_0^{(1)}$ and $\delta_0^{(2)}$ be the initial AoI of user 1 and user 2 at the beginning of this Type-2 chunk. Without loss of generality, we assume $\delta_0^{(2)}>\delta_0^{(1)}$. Then, user 2 will be updated first. Assume the generation times and delivery times of the two updates under the original resource chunk allocation are $t_1<t_2<d_2\leq d_1$. Then, users 1 and 2 will be updated at times $d_1$ and $d_2$, respectively. Besides, based on Lemma~\ref{lemma:chunk}, we have $d_2-t_2=i+1$. After the re-allocation, the DoF will be utilized to update user 2 first, starting at time $t_1$. Denote the corresponding delivery time as $d'_1$. Then, by the assumption that the transmission time of the first delivered update after reallocation is $i+1$, we have $d'_1-t_1=i+1$. Let $t'_2$ be the generation time of the second delivered update after re-allocation. Then, $t_1<t_2\leq d'_1\leq t'_2\leq d_2<d_1$.

\begin{figure}[t]
	\centering
	\includegraphics[width=2.7in]{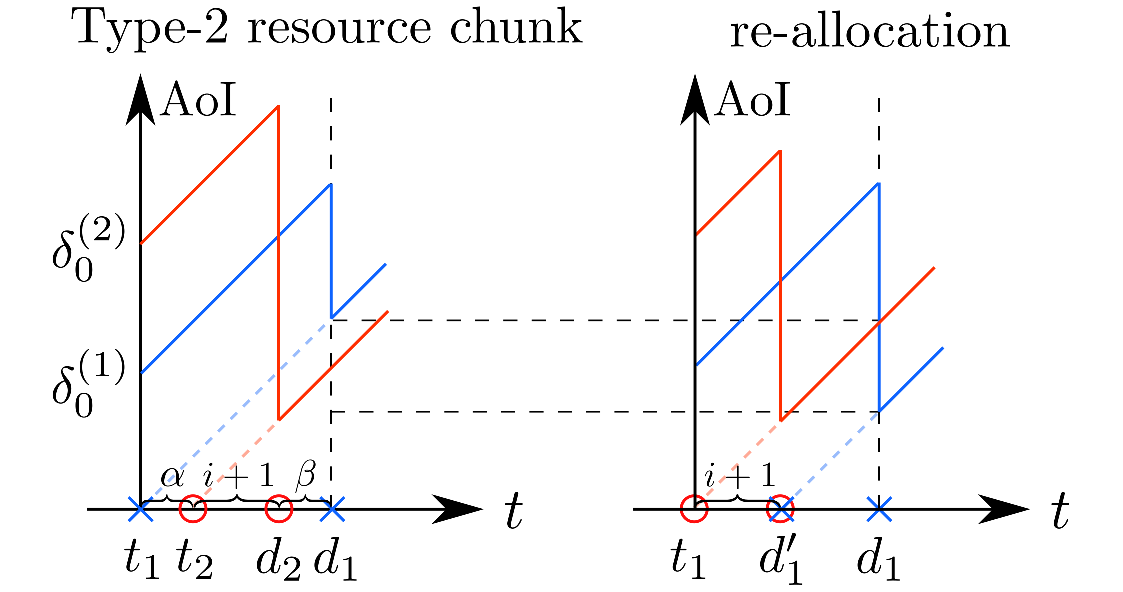}
	\captionof{figure}{Re-allocation of Type-2 resource chunk.} 
	\label{fig:type2-chunk}
\end{figure}

For clarity, we define $\alpha:=t_2-t_1\leq i+1$ and $\beta:=d_1-d_1\leq d'_2-t'_2\leq  i+2$. As illustrated in Fig.~\ref{fig:type2-chunk}, the AoI evolution before $d'_1$ stays the same after re-allocation. Besides, at time $d_1$, under the original allocation, the AoIs at users 1 and 2 are $d_1-t_1$, $d_1-t_2$ respectively. After re-allocation, the AoIs become $d_1-t'_2$, $d_1-t_1$, respectively. Since the DoF allocation after this resource chunk will be switched between the two users, user 2 will be updated next after the re-allocation. Since $d_1-t'_2\leq d_1-t_2$. the AoI evolution after $d_2$ will be improved after the re-allocation. It remains to show that the AoI over time slots $[d'_1:d_1-1]$ after the re-allocation is strictly improved.

Since the AoI evolution of user 1 stays the same before $d_1$ and the difference solely depends on the AoI of user 2, the age difference between that under the Type-2 resource chunk allocation and the re-allocation is
\begin{align*}
&\left(\sum_{\ell=1}^{\alpha}(\delta_0^{(2)}+i+\ell)+\sum_{\ell=1}^{\beta-1}(i+\ell)\right)-\sum_{\ell=1}^{\alpha+\beta-1}(i+\ell) \nonumber \\
&=\sum_{\ell=1}^{\alpha}\delta_0^{(2)}-\sum_{\ell=1}^{\beta-1}\alpha=\alpha(\delta_0^{(2)}-\beta+1)\geq 0,
\end{align*}
where the last inequality is based on the fact that the minimum possible AoI of the user with higher AoI is $i+1+\lceil \frac{2}{j}\rceil$ by Lemma~\ref{lemma:thm2-case5:minAoI}, and the fact that $\beta\leq i+2$.
\end{proof}

For a generalized block $\tilde{B}_{u,v}$, after all Type-2 resource chunks are re-allocated, it becomes a block $B_{u,v}$, which can be partitioned into segments as in Sec.~\ref{sec:converse-thm2}. Recall that each segment consists of either $ij+1$ or $(j+1)i+1$ time slots. Besides, for any policy in $\Pi_1$, the transmission time of any update is either $i+1$ or $i+2$.

Note that the transmission times of all updates in the segments of length $ij+1$ in $B_{u,v}$ are $i+1$. Then, according to Lemma~\ref{lemma:thm2-chunk_i+1}, under the original optimal updating scheme, it must not contain any Type-2 resource chunk before the first updating time slot in the next segment. The only possible segments that contain Type-2 resource chunks under the original optimal policy are segments of length $(j+1)i+1$ where the transmission time of the first update is $i+2$ and that of any other update is exactly $i+1$. Thus, under the original optimal policy, the Type-2 resource chunk can only be used to transmit the first two updates in the segment. In the following, we will show that the lower bound $\Delta_{\min}$ suggested in Lemma~\ref{lemma:thm2-case3:seg1} still holds for such segments.

\begin{table*}
\caption{Minimum AoI pattern of the first two updates for segments of length $(j+1)i+1$, $j\geq 2$. The first delivered updates is generated at the end of time slot $(U_m-1)i+m+\gamma$, and the two updates are delivered at $U_m i+m+\gamma$, $(U_m+1)i+m$, respectively.}
\centering
{\small
\begin{tabular}{|c|c|c|c|c|}\hline
Time slot  & $(U_m-1)i+m$    & $\cdots$ & $(U_m-1)i+m+\gamma-1$  & $(U_m-1)i+m+\gamma$    \\ \hline
Minimum AoI pattern & $(i+2,2i+2)$  & $\cdots$ & $(i+\gamma+1,2i+\gamma+1)$ & $(i+\gamma+2,2i+\gamma+2)$ \\  \hline\hline
Time slot   & $(U_m-1)i+m+\gamma+1$  & $\cdots$ & $U_m i+m+\gamma-1$   & $U_m i+m+\gamma$    \\ \hline
Minimum AoI pattern  & $(i+\gamma+3,2i+\gamma+3)$ & $\cdots$ & $(2i+\gamma+1,3i+\gamma+1)$ & $(i+1,2i+\gamma+2)$ \\  \hline\hline
Time slot    & $U_m i+m+\gamma+1$  & $\cdots$ & $(U_m+1)i+m-1$   & $(U_m+1)i+m$   \\ \hline
Minimum AoI pattern  & $(i+2,2i+\gamma+3)$ & $\cdots$ & $(2i-\gamma,3i+1)$  & $(2i-\gamma+1,2i+2)$  \\  \hline
\end{tabular}
}
\label{tab:new-1}
\end{table*}

\begin{table*}
\caption{Minimum AoI pattern of for segments of length $2i+1$. The first delivered update is generated at the end of time slot $(U_m-1)i+m+\gamma$, and the two updates are delivered at $U_m i+m+\gamma$, $(U_m+1)i+m$, respectively.}
\centering
{\small
\begin{tabular}{|c|c|c|c|c|}\hline
Time slot  & $(U_m-1)i+m$    & $\cdots$ & $(U_m-1)i+m+\gamma-1$  & $(U_m-1)i+m+\gamma$    \\ \hline
Minimum AoI pattern & $(i+2,2i+3)$  & $\cdots$ & $(i+\gamma+1,2i+\gamma+2)$ & $(i+\gamma+2,2i+\gamma+3)$ \\  \hline\hline
Time slot   & $(U_m-1)i+m+\gamma+1$  & $\cdots$ & $U_m i+m+\gamma-1$   & $U_m i+m+\gamma$    \\ \hline
Minimum AoI pattern  & $(i+\gamma+3,2i+\gamma+4)$ & $\cdots$ & $(2i+\gamma+1,3i+\gamma+2)$ & $(i+1,2i+\gamma+2)$ \\  \hline\hline
Time slot    & $U_m i+m+\gamma+1$  & $\cdots$ & $(U_m+1)i+m-1$   & $(U_m+1)i+m$   \\ \hline
Minimum AoI pattern  & $(i+2,2i+\gamma+3)$ & $\cdots$ & $(2i-\gamma,3i+1)$  & $(2i-\gamma+1,2i+2)$  \\  \hline
\end{tabular}
}
\label{tab:new-2}
\end{table*}

\begin{Lemma}\label{lemma:thm2-chunk-lowerbound}
For segments of length $(j+1)i+1$ where the first two updates are transmitted using a Type-2 resource chunk under the original optimal policy, the summed average AoI is lower bounded by $\Delta_{\min}=4i+1+\frac{2i+1}{ij+1}$ if $j\in\Zb_{\geq 2}$, and by $\Delta_{\min}=4i+3$ if $j=1$.
\end{Lemma}
\begin{proof}
\emph{1) $j\geq 2$.} Consider a segment of length $(j+1)i+1$ by $U_m$ that consists of time slots $[(U_m-1)i+m:(U_m+j)i+m]$. Following similar arguments as in the proof of Lemma~\ref{lemma:thm2-case3:seg1}, we can show that update $U_m$ is generated at time $(U_m-1)i+m-1$ instead of $(U_m-1)i+m$. Note that under the original resource allocation, within the Type-2 resource chunk, update $U_m$ is delivered after update $U_m+1$. Assume $U_m+1$ is generated at time $(U_m-1)i+m+\gamma$, $\gamma\in[0:i]$. It will consume all DoF until it is delivered at time $(U_m-1)i+m+\gamma+i$. Then the minimum AoI pattern of the first two updates over the segment can be specified (cf.~Table~\ref{tab:new-1}). At the end of time slot $(U_m+1)i+m$, the minimum AoI pattern is $(2i-\gamma+1,2i+2)$ under the original resource allocation and $(i+1,2i+2)$ under the re-allocation (cf.~Table~\ref{tab:thm2:2}). Since $2i-\gamma+1\geq i+1$ and the remaining updating follows the alternating updating rules, it suffices to show that the summed AoI over $[(U_m-1)i+m:(U_m+1)i+m]$ under the original allocation is greater than that under the re-allocation. In fact, the summed AoI over  $[(U_m-1)i+m:(U_m+1)i+m]$ under the original allocation is $8i^2+9i+3+\gamma^2+3$ (cf.~Table~\ref{tab:new-1}) and that under the re-allocation  (cf.~Table~\ref{tab:thm2:2}) is $8i^2+9i+3$, which completes the proof.

\emph{2) $j=1$.} The Type-2 resource chunk has length $2i+1$. According to Remark~\ref{remark:min}, the minimum AoI pattern is $(i+1,2i+2)$ instead of $(i+2,2i+1)$. Therefore, the minimum AoI pattern at the first time slot of the segment is $(i+2,2i+3)$ instead of $(i+2,2i+2)$ (cf.~\ref{tab:new-2}). Similar to 1), we assume that update $U_m+1$ is generated at time $(U_m-1)i+m+\gamma$, $\gamma\in[0:i]$, and the AoI pattern with the original Type-2 resource allocation can be specified (cf.~Table~\ref{tab:new-2}). The corresponding summed average AoI is $4i+3+\frac{\gamma^2+(i+1)\gamma}{2i+1}$, which is greater than or equal to $4i+3$, the summed average AoI under the re-allocation.
\end{proof}

In summary, the summed average AoI of any segment of the generalized block $\tilde{B}_{u,v}$ is lower bounded by $\Delta_{\min}$. Together with Lemma~\ref{lemma:idling}, we can show that the summed average AoI of the generalized block $\tilde{B}_{u,v}$ is also bounded by $\Delta_{\min}$. Therefore, the summed average AoI under any policy in $\tilde{\Pi}_{0}$ is lower bounded by $\Delta_{\min}$.

\subsection{Converse of Theorem~\ref{result2_new} (iv)}\label{app:result2_iv}
In this subsection, we provide a proof of the converse of Theorem~\ref{result2_new} (iv). I.e., for the case when $N\leq \frac{M}{2}$, $\frac{j}{ij+1}\leq \frac{N}{B}<\frac{j+1}{(j+1)i+1}$ where $i=\lceil \frac{B}{N}\rceil -1$ and $j=\lfloor \frac{1}{B/N-i}\rfloor$, we aim to show that the minimum summed average AoI is lower bounded by $3i+1+\frac{i+1}{ij+1}$ if $j\geq 1$, and by $3i-1$ if $j=0$. 
We adopt a similar approach as in the proof of the converse of Theorem~\ref{result2_new} (iii). For the sake of completeness and brevity, we provide key steps and omit proofs of lemmas if they are essentially the same as their counterparts in Section~\ref{sec:converse-thm2}.

Since the number of transmit antennas $M$ is more than the total number of antennas at the two users, each user is able to decode $N$ symbols simultaneously and the minimum summed average AoI is twice the minimum average AoI of the single-user system $(1,N,N,B)$. Thus, in the following, we will focus on the minimum average AoI of the single user system $(1,N,N,B)$.

Since the optimal scheme is guaranteed to lie within $\Pi_0$ by Theorem~\ref{thm:deterministic_policy}, any policy $\pi\in\Pi_0$ in $(1,N,N,B)$ systems can be represented as a sequence of blocks, where each block $B_{u,v}$ consists of $v$ idling time slots followed by $l_u:=\lceil uB/N \rceil$ time slots, during which the transmitter exhausts all DoF to send $u$ updates. When $v=0$, we express $B_{u,0}$ as $B_u$.

We first study a work-conserving updating scheme $\pi\in\Pi_0$, under which the transmitter exhausts its DoF at each time slot. 
\begin{Lemma}\label{lemma:converse:result3:1}
For the $(1,N,N,B)$ system with $\frac{j}{ij+1}\leq\frac{N}{B}<\frac{j+1}{(j+1)i+1}$, $i,j\in \Nb$, we have
\begin{itemize}
\item[(i)] the minimum AoI pattern in any time slot is $i+1$;
\item[(ii)] under policy $\pi$, the duration between two consecutive delivered updates is either $i$ or $i+1$ time slots.
\end{itemize}
\end{Lemma}
The proof of Lemma~\ref{lemma:converse:result3:1} is similar to that in Lemma~\ref{lemma:thm2-case5:minAoI} and Lemma~\ref{lemma:thm2-duration} and omitted.

Label the delivered updates starting at time 1 in the order of their delivery time. Let $U_m$ be the index of the $m$-th update whose delivery time is $i+1$ time slots after the previous delivered update. By the DoF constraint, we have
\begin{align}
&[i(U_m-1)+m-1]N\geq (U_m-1)B, \\
&[iU_m+m-1]N\geq U_m B.
\end{align}
Since $\frac{j}{ij+1}\leq\frac{N}{B}<\frac{j+1}{(j+1)i+1}$, $i,j\in \Nb$, solving the above inequalities gives
\begin{align}
j\hspace{-0.03in}-\hspace{-0.03in}1\hspace{-0.03in}<\hspace{-0.03in}\frac{N/B}{1\hspace{-0.03in}-\hspace{-0.03in}iN/B}\hspace{-0.03in}-\hspace{-0.03in}1\hspace{-0.03in}<\hspace{-0.03in}U_{m+1}\hspace{-0.03in}-\hspace{-0.03in}U_m\hspace{-0.03in}<\hspace{-0.03in}\frac{N/B}{1\hspace{-0.03in}-\hspace{-0.03in}iN/B}\hspace{-0.03in}+\hspace{-0.03in}1\hspace{-0.03in}<\hspace{-0.03in}j\hspace{-0.03in}+\hspace{-0.03in}2.
\end{align}
Since $U_m$ and $U_{m+1}$ are integers, we have
\begin{align}
U_{m+1}-U_m=j\quad \mbox{or} \quad j+1.   \label{eqn:converse:result3:1}
\end{align}

Partition the time axis into segments by the delivery time of updates $\{U_m-1\}_{m=2}^{\infty}$, i.e., $[1:ij+1]$, $\ldots$, $[(U_m-1)i+m:(U_{m+1}-1)i+m]$, $\ldots$. By Eqn.~(\ref{eqn:converse:result3:1}), the segment length is either $ij+1$ or $(j+1)i+1$. An example of the definition of $U_m$ and the segments for the $(1,7,7,12)$ system is shown in Fig.~\ref{fig:segment_ex2}. We point it out that all updates in Fig.~\ref{fig:segment_ex2} are intended for one user only while in Fig.~\ref{fig:segment_ex}, updates in the first row are intended for the first user and updates in the second row are intended for the second user.

\begin{figure}[t]
	\centering
	\includegraphics[width=3.4in]{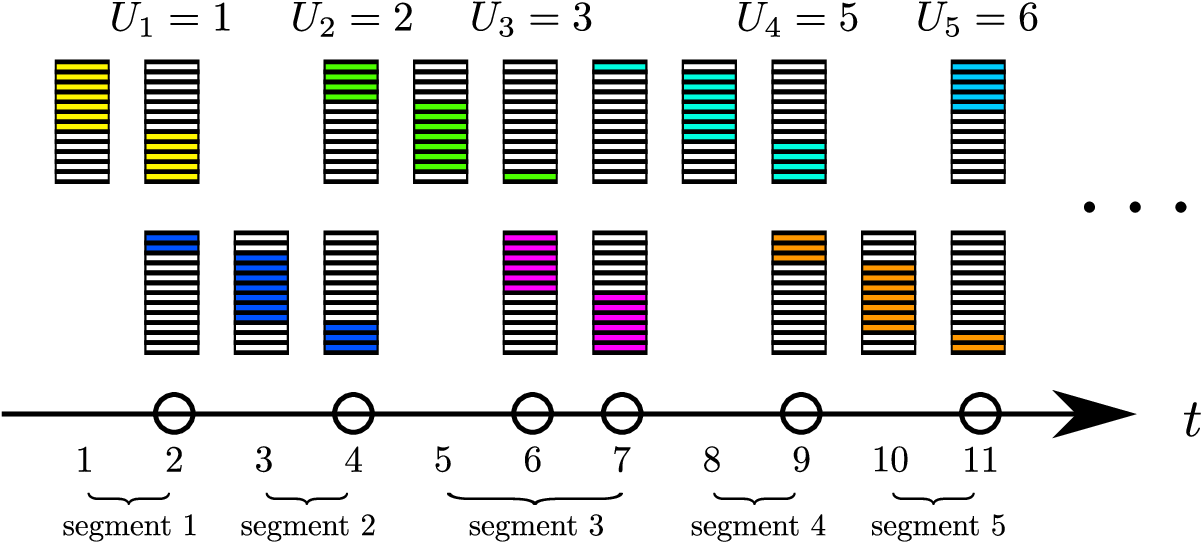}
	\captionof{figure}{Updating pattern in the $(1,7,7,12)$ system under $\pi_1$. We have $i=1$, $j=1$. Circles represent delivery times of updates. Since $i=1$, the segments can be obtained by tracking the updates whose delivery time is 2 time slots after the previous delivery time. We note that the length of each segment is either 2 (i.e., $ij+1$) or 3 (i.e., $(j+1)i+1$). }
	\label{fig:segment_ex2}
\end{figure}

\begin{Lemma}\label{lemma:converse:result3:2}

For the $(1,N,N,B)$ system with $\frac{j}{ij+1}\leq\frac{N}{B}<\frac{j+1}{(j+1)i+1}$, $i,j\in \Nb$, under policy $\pi$, the summed average AoI over segment $[(U_m-1)i+m:(U_{m+1}-1)i+m]$ is lower bounded by $\Delta_{\min}=\frac{3i+1}{2}+\frac{i+1}{2(ij+1)}$.

\end{Lemma}
\begin{table*}
\caption{Minimum AoI pattern for segments of length $ij+1$, $j\in\Nb$, $\ell\in[1:j-1]$. An update is delivered at the end of the time slots in the last column.}
\centering
{\small
\begin{tabular}{|c|c|c|c|c|}\hline
Time slot  & $(U_m-1)i+m$   & $\cdots$ & $U_m i+m-1$ & $U_m i+m$     \\ \hline
Minimum AoI pattern & $i+2$ & $\cdots$ & $2i+1$ & $i+1$ \\  \hline\hline
Time slot    & $(U_m+\ell-1)i+m+1$  & $\cdots$ & $(U_m+\ell)i+m-1$   &  $(U_{m}+\ell)i+m$   \\ \hline
Minimum AoI pattern & $i+2$ & $\cdots$ & $2i$ & $i+1$  \\ \hline
\end{tabular}
}
\label{tab:converse:result3:1}
\end{table*}
\begin{proof}
{\it 1)} The segment length is $ij+1$, i.e., $U_{m+1}-U_m=j$. The minimum AoI pattern over the segment can be specified in Table~\ref{tab:converse:result3:1} and the corresponding minimum average AoI is $\frac{3i+1}{2}+\frac{i+1}{2(ij+1)}$.

{\it 2)} The segment length is $(j+1)i+1$, i.e., $U_{m+1}-U_m=j+1$. By a similar argument as in Lemma~\ref{lemma:thm2-case3:seg1}, we can show that the minimum AoI pattern when $U_m$ is delivered, i.e., at the end of time slot $U_m i+m$, is $i+2$ instead of $i+1$. As a result, the minimum AoI pattern over segment of length $(j+1)i+1$ can be specified in Table~\ref{tab:converse:result3:2} and the minimum average AoI over the segment is $\frac{3i+1}{2}+\frac{3i+1}{2[(j+1)i+1]}$, which is greater than $\frac{3i+1}{2}+\frac{i+1}{2(ij+1)}$.

Combining the two cases, the average AoI over any segment is lower bounded by $\Delta_{\min}$.
\end{proof}
\begin{table*}
\caption{Minimum AoI pattern for segments of length $(j+1)i+1$, $j\in\Nb$, $\ell\in[2:j]$. An update is delivered at the end of the time slots in the last column.}
\centering
{\small
\begin{tabular}{|c|c|c|c|c|}\hline
Time slot  & $(U_m-1)i+m$    & $\cdots$ & $U_m i+m-1$  & $U_m i+m$    \\ \hline
Minimum AoI pattern & $i+2$  & $\cdots$ & $2i+1$ & $i+2$ \\  \hline\hline
Time slot   & $U_m i+m+1$  & $\cdots$ & $(U_m+1)i+m-1$   & $(U_m+1)i+m$    \\ \hline
Minimum AoI pattern  & $i+3$ & $\cdots$ & $2i+1$ & $i+1$ \\  \hline\hline
Time slot    & $(U_m+\ell-1)i+m+1$  & $\cdots$ & $(U_m+\ell)i+m-1$ & $(U_{m}+\ell)i+m$       \\ \hline
Minimum AoI pattern  & $i+2$ & $\cdots$ & $2i$ & $i+1$  \\  \hline
\end{tabular}
}
\label{tab:converse:result3:2}
\end{table*}

\begin{remark}\label{remark:converse:result3:1}
Note that for all $i,j\in \Nb$, the minimum average AoI over the first $\ell i+1$ time slots in each segment $[(U_m-1)i+m:(U_{m+1}-1)i+m]$ is monotonically decreasing in $\ell$ for $\ell\in[1:U_{m+1}-U_m]$.
\end{remark}
Next, we relate the AoI pattern under $\pi$ with block $B_u$. The updating policy $\pi$ over $[1,l_u]$ is identical to a block $B_u$ except that the DoF at time slot $l_u$ may not be exhausted. Partition $B_u$ into segments $[1:ij+1]$, $\ldots$, $[(U_{m_u-1}-1)i+m_u-1:(U_{m_u}-1)i+m_u-1]$ and a residue $[(U_{m_u-1}-1)i+m_u:l_u]$, where $m_u=\max\{m:U_m<l_u\}$. According to Remark~\ref{remark:converse:result3:1}, the average AoI of the residue is lower bounded by $\Delta_{\min}$.

For block $B_{u,v}$, the lower bound $\Delta_{\min}$ still holds as summarized in the following lemma.
\begin{Lemma}\label{lemma:converse:result3:3}
For the $(1,N,N,B)$ system with $\frac{j}{ij+1}\leq\frac{N}{B}<\frac{j+1}{(j+1)i+1}$, $i,j\in \Nb$, the summed average AoI over $B_{u,v}$ is lower bounded by $\Delta_{\min}=\frac{3i+1}{2}+\frac{i+1}{2(ij+1)}$. 
\end{Lemma}
\begin{proof}
Assume $B_{u,v}$ starts at time slot 1. Let $\delta_0$ be the AoI at time zero, $v+t$ be the first update delivery time and $v+l_u$ be the end of block $B_{u,v}$. Then, the existence of idling time slots affects the AoI evolution until $v+t$. By Lemma~\ref{lemma:converse:result3:1}, $\delta_0\geq i+1$ and $t\geq i+1$.

Let $A_1$ be the AoI increment induced by the idling time slots, as shown by the shaded area in Fig.~\ref{fig:idling2}. Denote $A_2$ be the summed AoI over $B_u$ when no idling time slot is present, corresponding to unshaded aread in Fig.~\ref{fig:idling2}. We have
\begin{align}
\frac{A_1}{v} &= \frac{1}{v}\left(\sum_{\ell=1}^v (\delta_0+\ell)+v\cdot t\right)=2i+2+\frac{v+1}{2} \nonumber \\
&\geq \frac{3i+1}{2}+\frac{i+1}{2(ij+1)}.
\end{align}
Let $\Delta_{B_{u,v}}$ be the summed average AoI over $B_{u,v}$. Then
\begin{align}
\Delta_{B_{u,v}}=\frac{A_1+A_2}{v+l_u}\geq \Delta_{\min},
\end{align}
where the last inequality follows from Lemma~\ref{lemma:ineq}.
\end{proof}

\begin{figure}[t]
	\centering
	\includegraphics[width=2.4in]{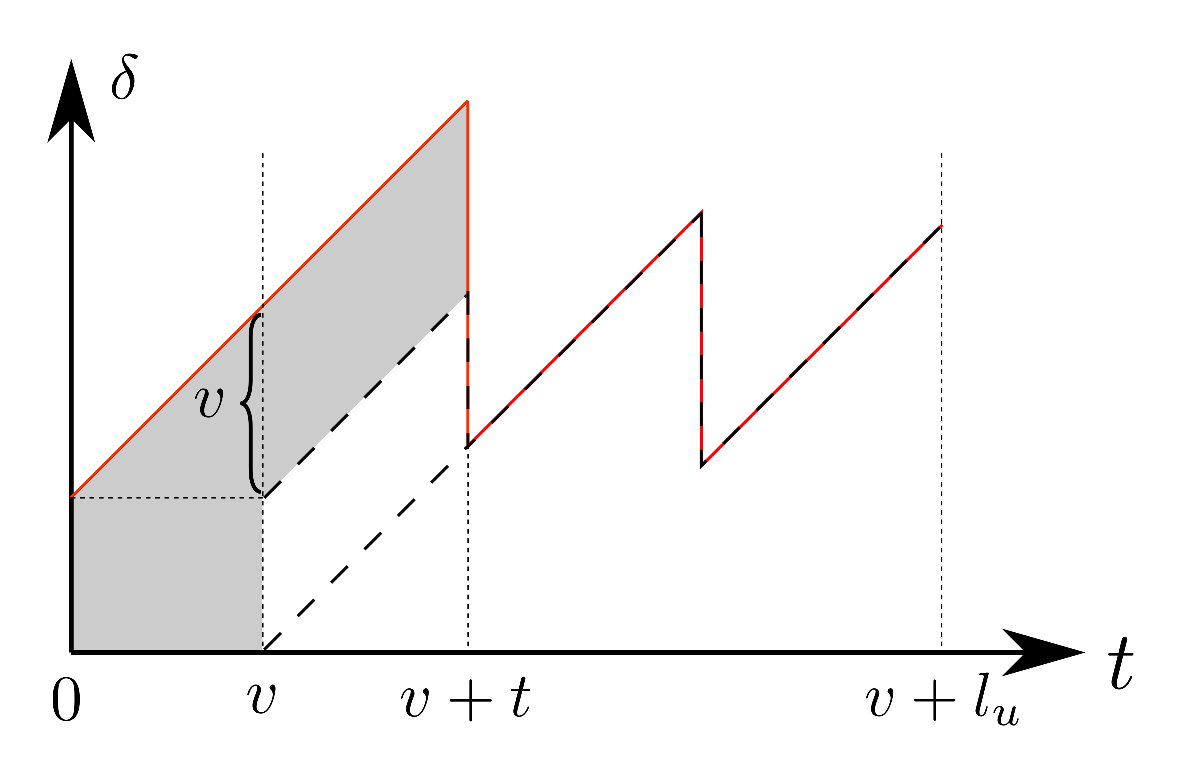}
	\captionof{figure}{AoI evolution over an extended block. The dashed line is the AoI evolution if idling period does not exist.} 
	\label{fig:idling2}
\end{figure}

We summarize the converse result in the following theorem.
\begin{Theorem}
For the $(2,M,N,B)$ system with $N\geq \frac{M}{2}$ and $M/B\in [\frac{j}{ij+1},\frac{j+1}{(j+1)i+1})$, $i,j\in\Nb$, the summed average AoI under any policy $\pi\in\Pi_0$ is lower bounded by $\Delta_{\min}=3i+1+\frac{i+1}{ij+1}$.
\end{Theorem}

\bibliographystyle{IEEEtran}
\bibliography{IEEEabrv,AgeInfo,ener_harv}

\begin{IEEEbiographynophoto}{Songtao Feng} received a B.S. degree in modern physics from the University of Science and Technology of China, Hefei, China in 2017. He is currently pursuing his PhD degree in electrical engineering at the Pennsylvania State University, University Park. His research interests include information theory, statistical learning, optimization and decision-making in wireless communications and networks.
\end{IEEEbiographynophoto}

\begin{IEEEbiographynophoto}{Jing Yang}(S'08-M'10) is an Associate Professor of Electrical Engineering at the Pennsylvania State University. She received her B.S. degree from the University of Science and Technology of China (USTC), and the M.S. and Ph.D. degrees from the University of Maryland, College Park, all in Electrical Engineering. She received the National Science Foundation CAREER award in 2015 and the WICE Early Achievement Award in 2020, and was selected as one of the 2020 N2Women: Stars in Computer Networking and Communications. She served as a Symposium/Workshop Co-chair for ICC 2021, INFOCOM 2021-AoI Workshop, WCSP 2019, CTW 2015, PIMRC 2014, a TPC Member of several conferences, and an Editor for IEEE TRANS. ON GREEN COMMUNICATIONS AND NETWORKING from 2017 to 2020. She is now an Editor for IEEE TRANS. ON WIRELESS COMMUNICATIONS. Her research interests lie in wireless communications and networking, information theory and statistical learning theory.

\end{IEEEbiographynophoto}

\end{document}